\documentclass[11pt,letterpaper,reqno]{amsart}

\title[Shape Dynamics of $N$ Point Vortices on the Plane]{Symplectic Reduction and the Lie--Poisson Shape Dynamics\\
  of $N$ Point Vortices on the Plane}
\author{Tomoki Ohsawa}
\address{Department of Mathematical Sciences, The University of Texas at Dallas, 800 W Campbell Rd, Richardson, TX 75080-3021}
\email{tomoki@utdallas.edu}
\date{\today}



\usepackage{graphicx,amssymb,paralist,subfigure,caption}
\usepackage[margin=1in, marginpar=.5in]{geometry}

\usepackage{tikz-cd}

\usepackage[numbers,sort&compress]{natbib}

\usepackage[colorlinks=false]{hyperref}

\theoremstyle{plain}
\newtheorem{theorem}{Theorem}[section]
\newtheorem{corollary}[theorem]{Corollary}
\newtheorem{lemma}[theorem]{Lemma}
\newtheorem{proposition}[theorem]{Proposition}

\theoremstyle{definition}

\theoremstyle{remark}
\newtheorem{remark}[theorem]{Remark}

\def\od#1#2{\dfrac{d#1}{d#2}}
\def\pd#1#2{\dfrac{\partial #1}{\partial #2}}
\def\tpd#1#2{\partial #1/\partial #2}

\def\parentheses#1{\!\left(#1\right)}
\def\brackets#1{\!\left[#1\right]}

\def\tr{\mathop{\mathrm{tr}}\nolimits}

\def\norm#1{\left\|#1\right\|}

\def\DS{\displaystyle}
\def\R{\mathbb{R}}
\def\C{\mathbb{C}}

\def\N{\mathbb{N}}
\def\defeq{\mathrel{\mathop:}=}
\def\eqdef{=\mathrel{\mathop:}}
\def\setdef#1#2{ \left\{ #1 \ |\ #2 \right\} }
\def\ip#1#2{\left\langle#1,#2\right\rangle}

\renewcommand{\Re}{\operatorname{Re}}
\renewcommand{\Im}{\operatorname{Im}}

\def\rmi{{\rm i}}

\def\d{\mathbf{d}}
\def\ins#1{{\bf i}_{#1}}
\def\PB#1#2{\left\{#1,#2\right\}}
\newcommand\Ad{\operatorname{Ad}}
\newcommand\ad{\operatorname{ad}}

\def\SO{\mathsf{SO}}
\def\SE{\mathsf{SE}}

\def\U{\mathsf{U}}

\def\u{\mathfrak{u}}
\def\so{\mathfrak{so}}

\newenvironment{tbmatrix}{\left[\begin{smallmatrix}}{\end{smallmatrix}\right]}

\begin{document}

\footskip=.6in

\begin{abstract}
  We show that the symplectic reduction of the dynamics of $N$ point vortices on the plane by the special Euclidean group $\mathsf{SE}(2)$ yields a Lie--Poisson equation for relative configurations of the vortices.
  Specifically, we combine symplectic reduction by stages with a dual pair associated with the reduction by rotations to show that the $\mathsf{SE}(2)$-reduced space with non-zero angular impulse is a coadjoint orbit.
  This result complements some existing works by establishing a relationship between the symplectic/Hamiltonian structures of the original and reduced dynamics.
  We also find a family of Casimirs associated with the Lie--Poisson structure including some apparently new ones.
  We demonstrate through examples that one may exploit these Casimirs to show that some shape dynamics are periodic.
\end{abstract}

\maketitle

\section{Introduction}
\subsection{Dynamics of $N$ Point Vortices}
The dynamics of $N$ point vortices $\{ \mathbf{x}_{j} = (x_{j},y_{j}) \in \R^{2} \}_{j=1}^{N}$ on the plane $\R^{2}$ with non-zero circulations $\{ \Gamma_{j} \in \R\backslash\{0\} \}_{j=1}^{N}$ is governed by the system of equations
\begin{equation}
  \label{eq:xy-HamSys}
  \dot{x}_{j} = -\frac{1}{2\pi} \sum_{\substack{1\le k \le N\\ k \neq j}} \Gamma_{k} \frac{y_{j} - y_{k}}{\norm{\mathbf{x}_{j} - \mathbf{x}_{k}}^{2}},
  \qquad
  \dot{y}_{j} = \frac{1}{2\pi} \sum_{\substack{1\le k \le N\\ k \neq j}} \Gamma_{k} \frac{x_{j} - x_{k}}{\norm{\mathbf{x}_{j} - \mathbf{x}_{k}}^{2}}
\end{equation}
for $j \in \{1, \dots, N\}$; see, e.g., \citet[Section~2.1]{Ne2001} and \citet[Section~2.1]{ChMa1993}.
This system of equations may be formulated as a Hamiltonian system as follows:
Let us equip $\R^{2N} = \{ (\mathbf{x}_{1}, \dots, \mathbf{x}_{N}) \}$ with the symplectic form
\begin{equation}
  \label{eq:Omega-xy}
  \Omega \defeq \sum_{j=1}^{N}\Gamma_{j} \d{x}_{j} \wedge \d{y}_{j}
\end{equation}
and define the Hamiltonian $H$ as
\begin{equation*}
  H(\mathbf{x}_{1}, \dots, \mathbf{x}_{N})
  \defeq -\frac{1}{4\pi} \sum_{1\le j < k \le N} \Gamma_{j} \Gamma_{k} \ln\norm{\mathbf{x}_{j} - \mathbf{x}_{k}}^{2}.
\end{equation*}
We note in passing that, strictly speaking, one needs to remove those collision points, i.e., those with $\mathbf{x}_{j} = \mathbf{x}_{k}$ with $j \neq k$, from $\R^{2N}$.
The vector field $X_{H}$ on $\R^{2N}$ defined by the Hamiltonian system $\ins{X_{H}}\Omega = \d{H}$ yields the above system of equations.
A common and more succinct way of describing the system is to identify $\R^{2}$ with $\C$ via $(x_{j},y_{j}) \mapsto x_{j} + \rmi y_{j} \eqdef q_{j}$ and write the symplectic form on $\R^{2N} \cong \C^{N} = \{ \mathbf{q} = (q_{1}, \dots, q_{N}) \}$ as
\begin{equation*}
  \Omega = -\frac{1}{2} \sum_{j=1}^{N} \Gamma_{j} \Im(\d{q}_{j} \wedge \d{q}_{j}^{*}) = -\d\Theta
\end{equation*}
with
\begin{equation*}
  \Theta \defeq -\frac{1}{2} \sum_{j=1}^{N} \Gamma_{j} \Im(q_{j}^{*} \d{q}_{j}),
\end{equation*}
and the Hamiltonian as
\begin{equation}
  \label{eq:H}
  H(q_{1}, \dots, q_{N}) = -\frac{1}{4\pi} \sum_{1\le j < k \le N} \Gamma_{j} \Gamma_{k} \ln|q_{j} - q_{k}|^{2}.
\end{equation}
Then the system is written as
\begin{equation}
  \label{eq:N-point_vortices}
  \dot{q}_{j} = \frac{\rmi}{2\pi} \sum_{\substack{1\le k \le N\\ k \neq j}} \Gamma_{k} \frac{q_{j} - q_{k}}{|q_{j} - q_{k}|^{2}}
\end{equation}
for $j \in \{1, \dots, N\}$.

This system has $\SE(2) = \SO(2) \ltimes \R^{2}$-symmetry under the action
\begin{equation}
  \label{eq:SE2-action}
  \SE(2) \times \C^{N} \to \C^{N};
  \qquad
  ((e^{\rmi\theta}, a), \mathbf{q}) \mapsto e^{\rmi\theta}\mathbf{q} + a\mathbf{1},
\end{equation}
where we identified $\R^{2}$ with $\C$ and defined $\mathbf{1} \defeq (1, \dots, 1) \in \C^{N}$.

It is well known (see, e.g., \citet[Equation~(2.1.5) on p.~69]{Ne2001}) that one may derive a set of equations for the inter-vortex distances $l_{ij} \defeq |q_{i} - q_{j}|$ of the point vortices; they are often referred to as the \textit{equations of relative motion} or the \textit{shape dynamics}; see \eqref{eq:relative_motion} below.
From the geometric point of view, this corresponds to the reduction of the dynamics by the above $\SE(2)$-symmetry:
This symmetry is essentially due to the uniformity of the ambient space, and hence ``dividing'' the dynamics by this symmetry results in the shape dynamics.
Such a reduction by symmetry---called symplectic or Hamiltonian reduction---is one of the main topics of the geometric approach to Hamiltonian dynamics; see, e.g., \citet{AbMa1978}, \citet{MaRa1999}, \citet{MaMiOrPeRa2007}, and references therein.
The use of shape space/dynamics is particularly popular in the $N$-body problem of classical mechanics; see, e.g., \citet{Iw1987}, \citet{Mo2015}, and references therein.

\subsection{Motivating Examples}
\label{ssec:motivating_examples}
We would like to show some motivating examples before discussing the main result of the paper.
The first example is the case with $N = 3$ and $\Gamma_{1} + \Gamma_{2} + \Gamma_{3} \neq 0$.
Using the inter-vortex distance $l_{ij} \defeq |q_{i} - q_{j}|$ between the vortices $i$ and $j$ and the signed area
\begin{equation*}
  A \defeq -\frac{1}{2}\Im( (q_{1} - q_{3})(q_{2}^{*} - q_{3}^{*})
  = \frac{1}{2}
  \begin{vmatrix}
    x_{1} - x_{3} & x_{2} - x_{3}\\
    y_{1} - x_{3} & y_{2} - x_{3}
  \end{vmatrix}.
\end{equation*}
of the triangle formed by the point vortices, we can derive, using \eqref{eq:N-point_vortices}, the \textit{equations of relative motion} mentioned above (see \citet[Equation~(2.1.5) on p.~69]{Ne2001}, \citet[Eqs.~(22) and (25)]{Ar2007}, and also references therein):
\begin{equation}
  \label{eq:relative_motion}
  \begin{array}{c}
    \DS \od{}{t}l_{jk}^{2} = \frac{2\Gamma_{i}}{\pi} \parentheses{ \frac{1}{l_{ki}^{2}} - \frac{1}{l_{ij}^{2}} } A
    \quad
    \text{for}
    \quad
    (i,j,k) \in \mathcal{Z}_{3},
    \\
    \DS \dot{A} = \frac{1}{8\pi} \sum_{(i,j,k) \in \mathcal{Z}_{3}} (\Gamma_{j}+\Gamma_{k}) \frac{l_{ki}^{2} - l_{ij}^{2}}{l_{jk}^{2}},
  \end{array}
\end{equation}
where $\mathcal{Z}_{3} \defeq \{ (1,2,3), (2,3,1), (3,1,2) \}$.
These equations govern the shape dynamics of the point vortices.

We will reformulate this set of equations as a Lie--Poisson equation on the dual of a certain Lie algebra, just as in \citet{BoPa1998} and \citet{BoBoMa1999} as a result of the $\SE(2)$-reduction.
The Lie--Poisson formulation helps us find a family of Casimirs (conserved quantities) including the following apparently new one:
\begin{equation}
  \label{eq:C_2}
  C_{2} = \sum_{(i,j,k) \in \mathcal{Z}_{3}} \parentheses{ \frac{l_{jk}^{4}}{\Gamma_{i}^{2}} + \frac{(l_{ij}^{2} - l_{jk}^{2} + l_{ki}^{2})^{2}}{2\Gamma_{j}\Gamma_{k}} } + 8\frac{\Gamma_{1} + \Gamma_{2} + \Gamma_{3}}{\Gamma_{1} \Gamma_{2} \Gamma_{3}} A^{2}.
\end{equation}

Figure~\ref{fig:3PVs} shows a numerical solution of \eqref{eq:xy-HamSys} along with the triangle connecting the initial positions of the three vortices with
\begin{equation}
  \label{eq:3PVs}
  (q_{1}(0), q_{2}(0), q_{3}(0)) = \parentheses{ 1 - 2\rmi, 2 + 4\rmi, -\frac{5}{3} - 2\rmi },
  \quad
  (\Gamma_{1}, \Gamma_{2}, \Gamma_{3}) = (5, 10, 15).
\end{equation}
This initial condition satisfies $\sum_{j=1}^{3} \Gamma_{j} q_{j}(0) = 0$, which is a conserved quantity called \textit{linear impulse} of the system as we shall see below in \eqref{eq:I}.
\begin{figure}
  \centering
  \includegraphics[width=.4\linewidth]{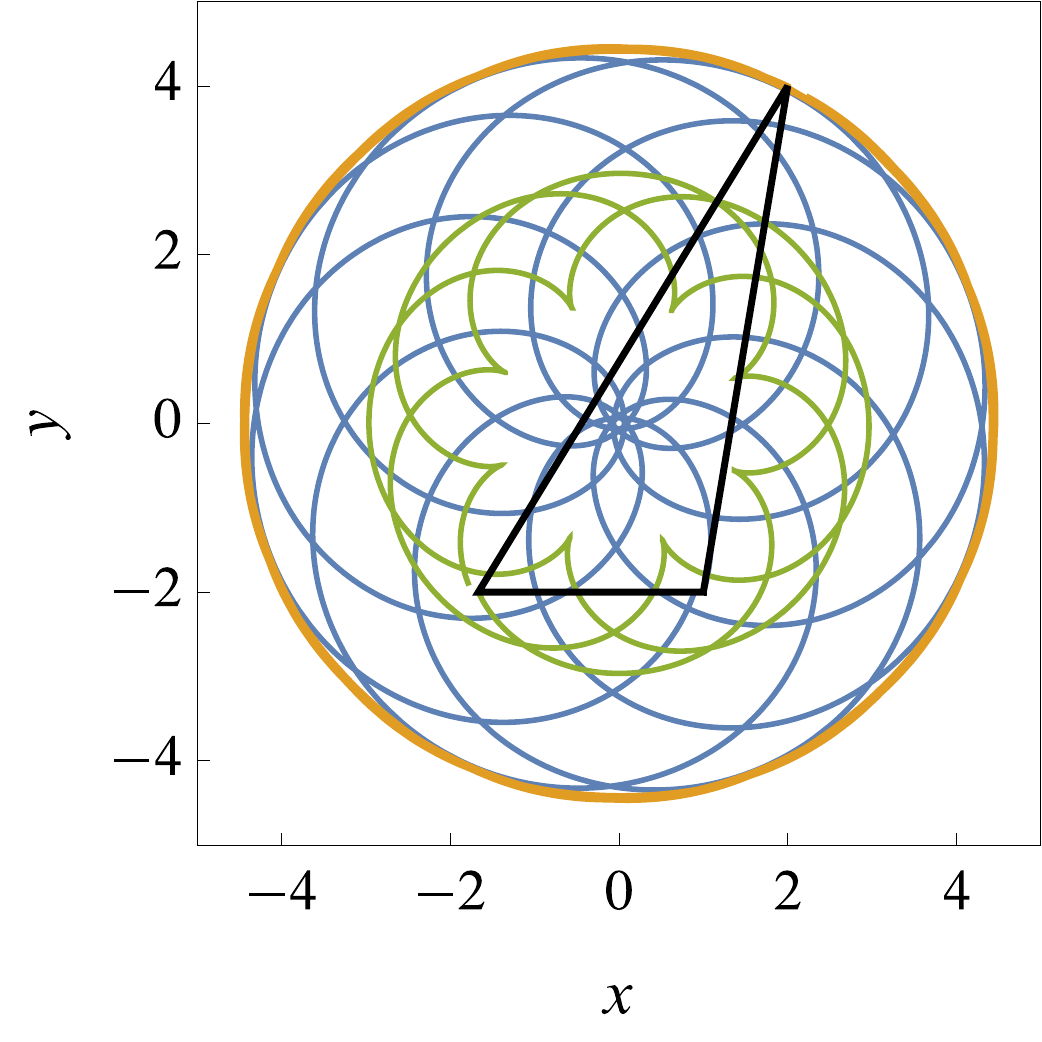}
  \captionsetup{width=0.95\textwidth}
  \caption{
    Numerical solution of \eqref{eq:xy-HamSys} with $N = 3$ and \eqref{eq:3PVs}.
    The black triangle is the shape formed by the \textit{initial} positions of the point vortices.
    Blue is the trajectory of the first vortex, orange the second, and green the third.
    The trajectories are \textit{not} exactly periodic.
  }
  \label{fig:3PVs}
\end{figure}
Figure~\ref{fig:Snapshots-3PVs} shows snapshots of the same solution along with the triangle connecting the positions of the point vortices.
\begin{figure}
  \centering
  \includegraphics[width=0.9\linewidth]{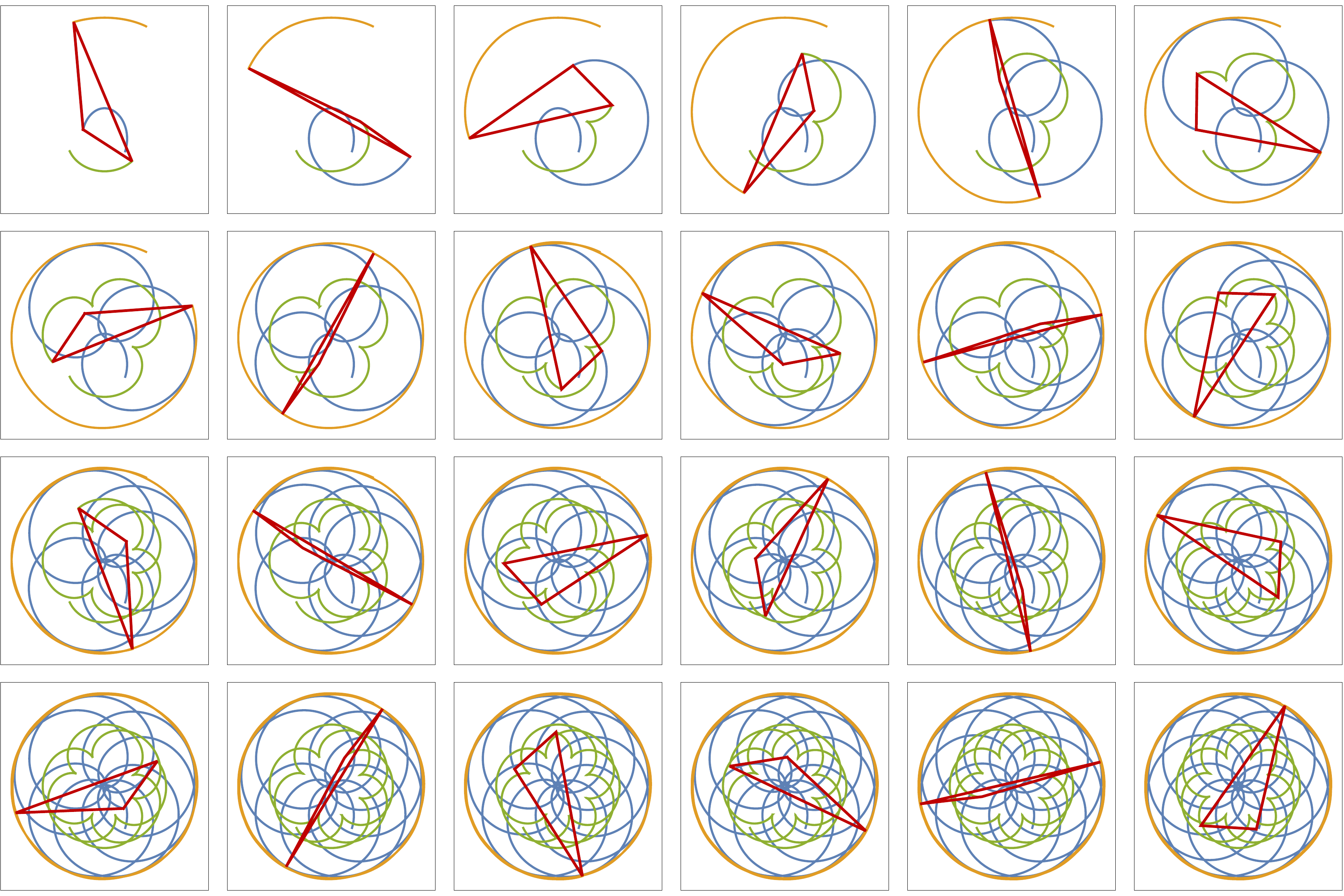}
  \captionsetup{width=0.95\textwidth}
  \caption{
    Snapshots of numerical solution from Fig.~\ref{fig:3PVs} with the triangle connecting the point vortices.
    The triangles in each column seem to be congruent, suggesting that the triangle changes its shape periodically.
  }
  \label{fig:Snapshots-3PVs}
\end{figure}
The triangle seems to change its shape periodically.
We will show in Section~\ref{ssec:N=3} that the shape dynamics is indeed periodic exploiting the Casimir~\eqref{eq:C_2}.
Note however that the \textit{trajectories} of the vortices are actually \textit{not} periodic:
These trajectories shown in Fig.~\ref{fig:3PVs} do not exactly come back to the initial positions. 

The other motivating example is the case with $N = 4$ and
\begin{equation}
  \label{eq:4PVs}
  (q_{1}(0), q_{2}(0), q_{3}(0), q_{4}(0)) = \parentheses{ 1 - 2\rmi, 2 + 4\rmi, 5, \frac{5}{8}(5 - \rmi) },
  \quad
  (\Gamma_{1}, \Gamma_{2}, \Gamma_{3}, \Gamma_{4}) = (5, 10, -7, -8)
\end{equation}
so that that $\sum_{j=1}^{4} \Gamma_{j} = 0$ as well as $\sum_{j=1}^{4} \Gamma_{j} q_{j}(0) = 0$.
Figure~\ref{fig:Snapshots-4PVs} shows a numerical solution of \eqref{eq:xy-HamSys} along with snapshots of the quadrangle connecting the positions of the point vortices.
The two neighboring quadrangles slightly left from the center at the bottom are the initial and terminal ones.
Note that the terminal quadrangle appears to be congruent to the initial shape, but is located at a slightly different position, again indicating that \textit{the shape dynamics may be periodic but the trajectories of the vortices are not}.
\begin{figure}
  \centering
  \includegraphics[width=0.4\linewidth]{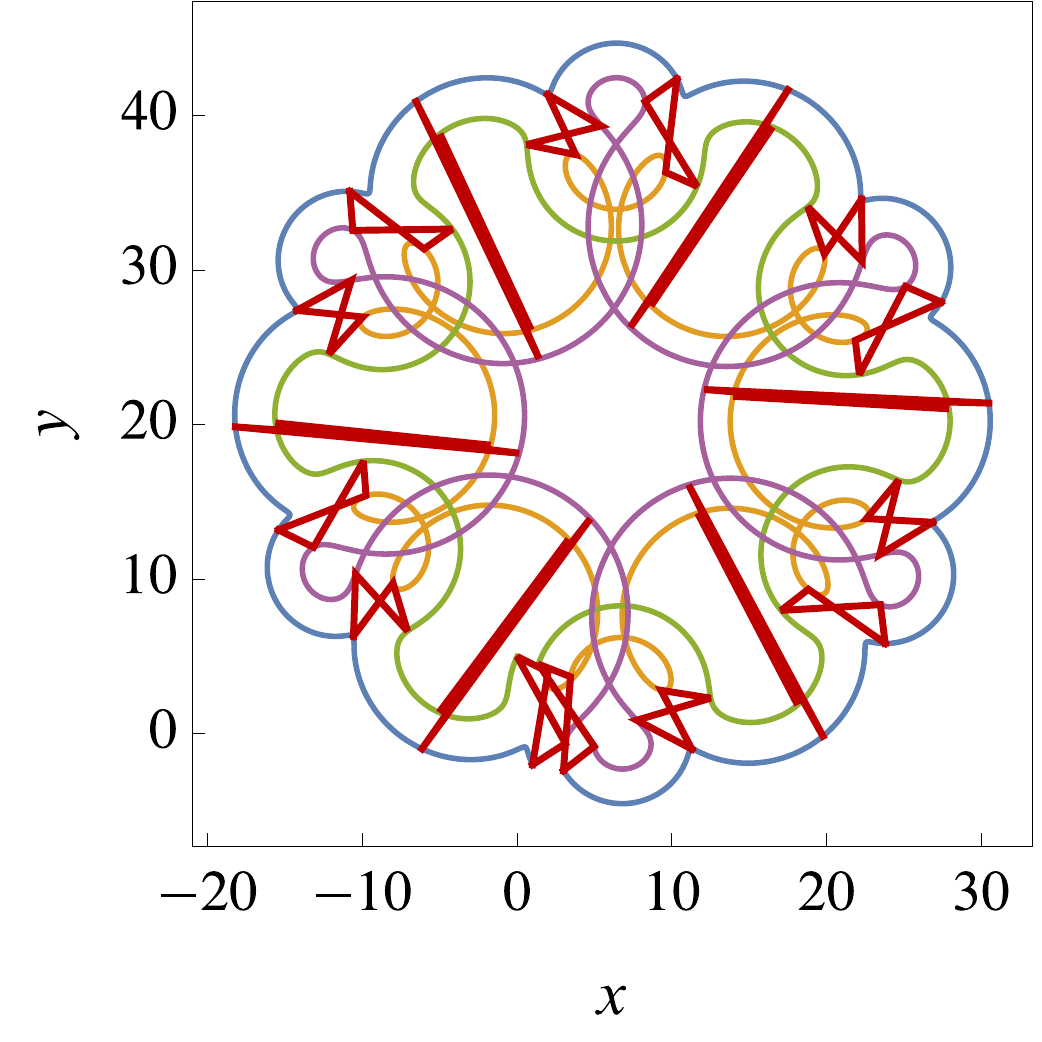}
  \captionsetup{width=0.95\textwidth}
  \caption{
    Numerical solution of \eqref{eq:xy-HamSys} with $N = 4$ and \eqref{eq:4PVs}.
    Blue is the trajectory of the first vortex, orange the second, green the third, and violet the fourth.
    It also shows snapshots of the quadrangle connecting the vortices in the numerical order; the shape dynamics again seems periodic.
  }
  \label{fig:Snapshots-4PVs}
\end{figure}

\subsection{Main Results}
We perform $\SE(2)$-reduction of the Hamiltonian dynamics of $N$ point vortices with non-zero angular impulse and show that the resulting dynamics can be written as a Lie--Poisson equation in a coadjoint orbit.
The main goal of this paper is to show that the $\SE(2)$-reduction naturally gives rise to the Lie--Poisson equation.

That one can write the reduced/shape dynamics of $N$ point vortices as a Lie--Poisson equation is not new.
\citet{BoPa1998} found the Lie--Poisson bracket for the reduced dynamics in a rather direct manner, and \citet{BoBoMa1999} gave a Lie-algebraic interpretation of the result by defining a so-called vortex algebra, and showed that it is isomorphic to the indefinite unitary algebra $\u(n_{1},n_{2})$ for some $n_{1}, n_{2} \in \{0, \dots, N-1\}$ such that $n_{1} + n_{2} = N-1$, depending on the signs of the circulations $\{ \Gamma_{j} \}_{j=1}^{N}$.
More recently, \citet{He2016} (see also \citet{HeSh2018}) showed that the reduced dynamics of three point vortices may be written as a Lie--Poisson equation on $\u(2)^{*}$ with the standard Lie--Poisson bracket by constructing a set of covectors satisfying the Pauli commutation relations.

The main contributions of this paper are the following:
\textit{
  (i)~We identify the Lie--Poisson structure as the natural symplectic structure on the reduced space by performing symplectic reduction by the $\SE(2)$-symmetry, thereby establishing a clear connection between the original symplectic structure~\eqref{eq:Omega-xy} with the Lie--Poisson structure (Theorem~\ref{thm:rotational_reduction}).
  The resulting Lie--Poisson equation yields the equations~\eqref{eq:relative_motion} of relative motion.
  (ii)~The Lie--Poisson structure naturally gives rise to Casimirs that may provide additional conserved quantities (Corollary~\ref{cor:reduced_dynamics}).
  Some of the Casimirs are apparently new, while others are well-known conserved quantities.
  We exploit these Casimirs to show that the shape dynamics from the above examples are in fact periodic.
}

\subsection{Outline}
\label{ssec:outline}
Particularly, we perform the $\SE(2)$-reduction by stages by first performing the reduction by $\R^{2}$ (see Section~\ref{sec:ReductionByTranslation}), and then by $\SO(2)$ (see Section~\ref{sec:ReductionByRotation}).
We note that \citet{BoBoMa1999} seem to work the other way around, i.e., first by rotations and then by translations, although it is not particularly clear how one can perform the $\R^{2}$-reduction of the $\SO(2)$-reduced space, nor how the symplectic structures are related to each other.
We stick to the former approach because that is the procedure justified by the semidirect product reduction (see, e.g., \citet[Theorem~4.2.2 on p.~122]{MaMiOrPeRa2007}).

Our work elucidates how the original symplectic structure $\Omega$ gives rise to a symplectic structure $\Omega_{Z}$ or $\Omega_{Z_{0}}$ (Propositions~\ref{prop:Omega_Z} and \ref{prop:Omega_Z_0}) on the $\R^{2}$-reduced space, and also in turn, $\Omega_{Z}$ or $\Omega_{Z_{0}}$ gives rise to the Lie--Poisson structure as a result of the $\SO(2)$-reduction if the angular impulse is non-zero (Theorem~\ref{thm:rotational_reduction}).
As we shall see in Section~\ref{sec:ReductionByTranslation}, the two symplectic structures $\Omega_{Z}$ and $\Omega_{Z_{0}}$ on the $\R^{2}$-reduced space correspond to those cases where the total circulation $\Gamma \defeq \sum_{j=1}^{N} \Gamma_{j}$ is non-zero and zero, respectively.
These two cases result in slightly different geometries and hence require separate treatments.
Nevertheless, the resulting symplectic structures $\Omega_{Z}$ and $\Omega_{Z_{0}}$ have similar structures, and hence the $\SO(2)$-reduction to follow works the same way.

As an aside, we note that the initial inspiration came from the work of \citet{Mo2015} on the reduction of the three-body problem (of celestial mechanics \textit{not} of point vortices).
The map $\Phi$ defined in~(25) (or $\pi^{\rm rot}$ defined in (31)) in \cite{Mo2015} used for reduction by rotational symmetry is a momentum map if one thinks of the \textit{configuration space} $\R^{2} \cong \C$---not its cotangent bundle---as a symplectic vector space in the standard manner.
While this symplectic structure on the configuration space has little significance in celestial mechanics, it is an essential ingredient in point vortex dynamics as its Hamiltonian formulation employs a variant~\eqref{eq:Omega-xy} of this symplectic structure.
The corresponding momentum map in our context constitutes one leg of the dual pair we will exploit in this paper; see Section~\ref{ssec:reduction_by_rotation}.

\section{Reduction by Translational Symmetry}
\label{sec:ReductionByTranslation}
The first stage of the $\SE(2)$-reduction by stages is the reduction by the translational symmetry.
As mentioned above, we need slightly different treatments depending on whether the total circulation $\Gamma \defeq \sum_{j=1}^{N} \Gamma_{j}$ is zero or not.

\subsection{Translational Symmetry and Momentum Map}
Consider the translational part of the $\SE(2)$-action \eqref{eq:SE2-action}, i.e., $\C \cong \R^{2}$-action on $\C^{N}$ as follows:
\begin{equation*}
  \C \times \C^{N} \to \C^{N};
  \qquad
  (a, \mathbf{q} \defeq (q_{1}, \dots , q_{N})) \mapsto \mathbf{q} + a\mathbf{1}.
\end{equation*}
The corresponding infinitesimal generator for $\alpha \in \C$ is then written as
\begin{equation*}
  \alpha_{\C^{N}}(\mathbf{q}) = \sum_{j=1}^{N} \parentheses{ \alpha\,\pd{}{q_{j}} + \alpha^{*}\pd{}{q_{j}^{*}} },
\end{equation*}
Then one sees that 
\begin{equation*}
  \ins{\alpha_{\C^{N}}}\Omega = \d{\mathcal{I}^{\alpha}}
\end{equation*}
with
\begin{align*}
  \mathcal{I}^{\alpha}(\mathbf{q})
  &\defeq -\frac{\rmi}{2} \sum_{j=1}^{N} \Gamma_{j} (\alpha^{*}q_{j} - \alpha q_{j}^{*}) \\
  &= \frac{1}{2}\brackets{ \parentheses{ -\rmi \sum_{j=1}^{N} \Gamma_{j} q_{j} }^{*} \alpha + \alpha^{*} \parentheses{ -\rmi \sum_{j=1}^{N} \Gamma_{j} q_{j}} } \\
  &= \ip{ -\rmi \sum_{j=1}^{N} \Gamma_{j} q_{j} }{\alpha}_{\C},
\end{align*}
where we defined an inner product on $\C$ as $\ip{\alpha}{\beta}_{\C} \defeq \Re(\alpha^{*}\beta)$.
Hence we have $\mathcal{I}^{\alpha}(\mathbf{q}) = \ip{ \mathbf{I}(\mathbf{q}) }{\alpha}$ with the momentum map $\mathbf{I}\colon \C^{N} \to \C^{*} \cong \C$ defined by
\begin{equation}
  \label{eq:I}
  \mathbf{I}(\mathbf{q}) \defeq -\rmi \sum_{j=1}^{N} \Gamma_{j} q_{j}.
\end{equation}
This is essentially the so-called \textit{linear impulse}; see, e.g., \citet[Section~2.1]{Ne2001} and \citet{Ar2007}.
By Noether's Theorem (see, e.g., \citet[Theorem~11.4.1]{MaRa1999}), this is a conserved quantity of the system~\eqref{eq:N-point_vortices}.

The above momentum map is \textit{not} equivariant except for a special case:
\begin{lemma}
  \label{lem:I-equivariance}
  The momentum map $\mathbf{I}$ is equivariant if and only if the total circulation
  \begin{equation*}
    \Gamma \defeq \sum_{j=1}^{N}\Gamma_{j}
  \end{equation*}
  vanishes.
\end{lemma}

\begin{proof}
  Since $\mathbb{C}$ is abelian, the coadjoint action is trivial; hence equivariance would be $\mathbf{I}(\mathbf{q} + a\mathbf{1}) = \mathbf{I}(\mathbf{q})$ for any $a \in \C$.
  However, it is straightforward to see that, for any $a \in \C$,
  \begin{equation*}
    \mathbf{I}(\mathbf{q} + a\mathbf{1}) = \mathbf{I}(\mathbf{q}) - \rmi \Gamma a. \qedhere
  \end{equation*}
\end{proof}

\subsection{Reduction by Translational Symmetry}
Let $c \in \C$ be arbitrary and consider the level set
\begin{equation}
  \label{eq:I-level_set}
  \mathbf{I}^{-1}(-\rmi c) = \setdef{ (q_{1}, \dots, q_{N}) \in \C^{N} }{ \sum_{j=1}^{N} \Gamma_{j}q_{j} = c },
\end{equation}
which defines an affine subspace of $\C^{N}$.
It has different symplectic-geometric properties depending on the value of the total circulation $\Gamma$:
\begin{lemma}
  The affine subspace $\mathbf{I}^{-1}(-\rmi c) \subset \C^{N}$ is symplectic if $\Gamma \neq 0$ whereas it is coisotropic if $\Gamma = 0$.
\end{lemma}

\begin{proof}
  Let us write $A \defeq \mathbf{I}^{-1}(-\rmi c)$ for short and find the symplectic orthogonal complement $(TA)^{\Omega}$ of the tangent space $TA$ of $A$.
  Let $\mathbf{q} \in A$ be arbitrary and $\mathbf{v} = (v_{1}, \dots, v_{N}) \in \C^{N}$ be an arbitrary element in $T_{\mathbf{q}}A$ by identifying $T_{\mathbf{q}}A$ with a subspace of $\C^{N}$ in a natural manner for notational simplicity.
  Then we have $\Gamma_{N} v_{N} = -\sum_{j=1}^{N-1}\Gamma_{j}v_{j}$.
  For an arbitrary $\mathbf{w} = (w_{1}, \dots, w_{N}) \in T_{\mathbf{q}}\C^{N}$, we have
  \begin{align*}
    \Omega(v, w)
    &= -\sum_{j=1}^{N} \frac{\Gamma_{j}}{2} \Im\parentheses{ v_{j} w_{j}^{*} - v_{j}^{*} w_{j} } \\
    &= \sum_{j=1}^{N} \Gamma_{j} \Im\parentheses{ v_{j}^{*} w_{j} } \\
    &= \Im\parentheses{ \sum_{j=1}^{N-1} \Gamma_{j} v_{j}^{*} (w_{j} - w_{N}) }.
  \end{align*}
  Since $v_{1}, \dots, v_{N-1} \in \C$ are arbitrary, it follows that
  \begin{equation*}
    (T_{\mathbf{q}}A)^{\Omega} = \setdef{ \mathbf{w} \in \C^{N} }{ w_{1} = \dots = w_{N} } = \C\mathbf{1},
  \end{equation*}
  where we defined
  \begin{equation*}
    \C\mathbf{1} \defeq \setdef{ a\mathbf{1} \in \C^{N} }{ a \in \C }.
  \end{equation*}
 Hence we see that
  \begin{equation*}
    T_{\mathbf{q}}A \cap (T_{\mathbf{q}}A)^{\Omega}
    = \setdef{ a\mathbf{1} \in \C^{N} }{ a \in \C,\, a\Gamma = 0 } \\
    =
    \begin{cases}
      \{0\} & \Gamma \neq 0, \\
      \C\mathbf{1} = (T_{\mathbf{q}}A)^{\Omega} & \Gamma = 0.
    \end{cases}
  \end{equation*}
  Therefore, if $\Gamma \neq 0$ then $A$ is symplectic, whereas if $\Gamma = 0$ then $(T_{\mathbf{q}}A)^{\Omega} \subset T_{\mathbf{q}}A$ for any $\mathbf{q} \in A$, and so $A$ is coisotropic.
\end{proof}

As a result, we obtain the reduced space as follows:
\begin{proposition}[Reduction by translational symmetry]\hfill
  \label{prop:translational_reduction}
  \begin{enumerate}[(i)]
  \item If $\Gamma \neq 0$, the reduced space by the translational symmetry is $\mathbf{I}^{-1}(-\rmi c)$ itself for any $c \in \C$; the affine subspace $\mathbf{I}^{-1}(-\rmi c)$ in turn may be identified with the subspace $\mathbf{I}^{-1}(0) \cong \C^{N-1}$.
  \item If $\Gamma = 0$, the reduced space is $\mathbf{I}^{-1}(-\rmi c)/\C$ and may be identified with $\mathbf{I}^{-1}(0)/\C \cong \C^{N-2}$.
  \end{enumerate}
\end{proposition}

\begin{proof}
  Suppose first that $\Gamma \neq 0$.
  By Lemma~\ref{lem:I-equivariance}, the momentum map $\mathbf{I}$ is not equivariant.
  Therefore, we would like to invoke the non-equivariant symplectic reduction (see, e.g., \cite[p.~17]{MaMiOrPeRa2007}).
  Based on what we observed in the proof of Lemma~\ref{lem:I-equivariance}, we define a cocycle $\sigma\colon \C \to \C^{*} \cong \C$ as
  \begin{equation*}
    \sigma(a) \defeq \mathbf{I}(\mathbf{q} + a\mathbf{1}) - \mathbf{I}(\mathbf{q}) = -\rmi \Gamma a.
  \end{equation*}
  This gives rise to the new action $\Xi\colon \C \times \C^{*} \to \C^{*}$ defined by
  \begin{equation*}
    \Xi(a,-\rmi c) \defeq -\rmi c + \sigma(a) = -\rmi(c + \Gamma a).
  \end{equation*}
  The isotropy group of this action is clearly trivial, i.e., $\C_{-\rmi c} = \{0\}$.
  Hence the (non-equivariant) Marsden--Weinstein quotient is $\mathbf{I}^{-1}(-\rmi c)$ itself.
  However, one may shift the origin of $\C^{N}$ so that the affine space $\mathbf{I}^{-1}(-\rmi c)$ becomes the subspace $\mathbf{I}^{-1}(0) \cong \C^{N-1}$.
  Note that this does not affect the dynamics because of the translational symmetry of the Hamiltonian~\eqref{eq:H}.

  Now suppose that $\Gamma = 0$.
  Then, by Lemma~\ref{lem:I-equivariance}, the momentum map $\mathbf{I}$ is equivariant.
  Since $\C$ is abelian, the isotropy group is given by $\C_{-\rmi c} = \C$.
  Hence we obtain the Marsden--Weinstein quotient $\mathbf{I}^{-1}(-\rmi c)/\mathbb{C}$.
  One sees from \eqref{eq:I-level_set} that $\mathbf{I}^{-1}(-\rmi c)$ defines an affine space of (complex) codimension one.
  Since $\C$ acts on it by translations in the direction of $\mathbf{1}$ inside $\mathbf{I}^{-1}(-\rmi c)$, one sees that the quotient $\mathbf{I}^{-1}(-\rmi c)$ is an affine space of (complex) codimension two, i.e., $\mathbf{I}^{-1}(-\rmi c)/\mathbb{C} \cong \C^{N-2}$.
  Alternatively, for the same reason as above, one may identify $\mathbf{I}^{-1}(-\rmi c)$ with the subspace $\mathbf{I}^{-1}(0)$.
  Then it is easy to see that $\mathbf{I}^{-1}(0)/\C$ is a quotient of a vector space $\mathbf{I}^{-1}(0) \cong \C^{N-1}$ by its subspace $\C\mathbf{1}$ and hence is isomorphic to $\C^{N-2}$.
  This is nothing but the linear symplectic reduction of a coisotropic subspace; see, e.g., \citet[Lemma~2.1.7]{McSa2016}.
\end{proof}

\subsection{Symplectic Forms on $\R^{2}$-Reduced Space}
Let us first consider the case with $\Gamma \neq 0$.
The above proposition tells us that the reduced space by translational symmetry may be identified with the subspace
\begin{equation*}
  \mathbf{I}^{-1}(0) = \setdef{ (q_{1}, \dots, q_{N}) \in \C^{N} }{ \sum_{j=1}^{N} \Gamma_{j}q_{j} = 0 }.
\end{equation*}
We parametrize this subspace using the relative positions of the first $N-1$ point vortices with respect to the last one, i.e.,
\begin{equation}
  \label{eq:z}
  z = (z_{1}, \dots, z_{N-1}) \defeq (q_{1} - q_{N}, \dots, q_{N-1} - q_{N}) \in \C^{N-1}.
\end{equation}
Then,
\begin{align*}
  \mathbf{I}^{-1}(0)
  &= \setdef{ (z_{1}, \dots, z_{N-1}, 0) + q_{N}\mathbf{1} \in \C^{N} }{ q_{N} = -\frac{1}{\Gamma} \sum_{j=1}^{N-1} \Gamma_{j} z_{j} } \\
  &\cong \{ (z_{1}, \dots, z_{N-1}) \in \C^{N-1} \} = \C^{N-1}.
\end{align*}
We remove the those points for $N$-tuple collisions $q_{1} = \dots = q_{N}$ or equivalently $z = 0$ to define
\begin{equation*}
  Z \defeq \mathbf{I}^{-1}(0) \backslash \{\text{$N$-tuple collisions}\}
  \cong \C^{N-1} \backslash \{0\}.
\end{equation*}
Let us find the symplectic form $\Omega_{Z}$ induced on $Z$ by $\Omega$.
\begin{proposition}
  \label{prop:Omega_Z}
  If $\Gamma \neq 0$, then the symplectic form on the $\R^{2}$-reduced space $Z$ can be written as
  \begin{equation*}
    \Omega_{Z} = -\d\Theta_{Z},
  \end{equation*}
  where $\Theta_{Z}$ is the one-form on $Z \cong \C^{N-1} \backslash \{0\}$ defined as
  \begin{equation*}
    \Theta_{Z} \defeq \frac{1}{2}\Im(z^{*} \mathcal{K} \d{z})
  \end{equation*}
  with
  \begin{equation}
    \label{eq:K}
    \mathcal{K} \defeq \frac{1}{\Gamma}
    \begin{bmatrix}
      -\Gamma_{1}(\Gamma - \Gamma_{1}) & \Gamma_{1} \Gamma_{2} & \dots & \Gamma_{1} \Gamma_{N-1} \\
      \Gamma_{2} \Gamma_{1} & -\Gamma_{2}(\Gamma - \Gamma_{2}) & \dots & \Gamma_{2} \Gamma_{N-1} \\
      \vdots & \vdots & \ddots & \vdots \\
      \Gamma_{N-1} \Gamma_{1} & \Gamma_{N-1} \Gamma_{2} & \dots & -\Gamma_{N-1}(\Gamma - \Gamma_{N-1})
    \end{bmatrix}.
  \end{equation}
\end{proposition}

\begin{proof}
  The constraint $\sum_{j=1}^{N}\Gamma_{j} q_{j} = 0$ for $q$ to be in $Z = \mathbf{I}^{-1}(0)$ is rewritten in terms of $z$ as
  \begin{equation*}
    \sum_{j=1}^{N}\Gamma_{j} q_{j} = 0
    \iff
    \sum_{j=1}^{N-1} \Gamma_{j} z_{j} + \Gamma q_{N} = 0
    \iff
    q_{N} = -\frac{1}{\Gamma} \sum_{j=1}^{N-1} \Gamma_{j} z_{j},
  \end{equation*}
  and thus we may write the embedding $\iota\colon Z \hookrightarrow \C^{N}$ as
  \begin{equation*}
    \iota\colon (z_{1}, \dots, z_{N-1}) \mapsto
    (z_{1} + q_{N}, \dots, z_{N-1} + q_{N}, q_{N}).
  \end{equation*}
  Then, straightforward calculations yield the pull-back
  \begin{align*}
    \Theta_{Z} &\defeq \iota^{*}\Theta \\
               &= \frac{1}{2\Gamma} \parentheses{
                 -\sum_{j=1}^{N-1} \Gamma_{j}(\Gamma - \Gamma_{j}) \Im(z_{j}^{*}\d{z_{j}})
                 + \sum_{\substack{1\le j, k\le N-1\\ j\neq k}} \Gamma_{j} \Gamma_{k} \Im(z_{j}^{*}\d{z_{k}})
                 } \\
               &= \frac{1}{2}\Im(z^{*} \mathcal{K} \d{z}).
  \end{align*}
  Hence the symplectic form on $Z$ is given by
  \begin{equation*}
    \Omega_{Z} = \iota^{*}\Omega = -\d(\iota^{*}\Theta) = -\d\Theta_{Z}. \qedhere
  \end{equation*}
\end{proof}

\begin{remark}
  The matrix $\mathcal{K}$ is invertible under our assumption that $\Gamma_{j} \neq 0$ for $j \in \{1, \dots, N\}$; see Lemma~\ref{lem:detK}.
\end{remark}

What if $\Gamma = 0$?
In this case, we may write the embedding $i_{0}\colon \mathbf{I}^{-1}(0) \hookrightarrow \C^{N}$ as
\begin{equation*}
  i_{0}\colon (q_{1}, \dots, q_{N-1}) \mapsto
  \parentheses{
    q_{1}, \dots, q_{N-1}, -\frac{1}{\Gamma_{N}} \sum_{j=1}^{N-1} \Gamma_{j} q_{j}
  }.
\end{equation*}

The pull-back of the canonical one-form $\Theta$ by $i_{0}$ is then
\begin{equation*}
  i_{0}^{*}\Theta =
  -\frac{1}{2\Gamma_{N}} \parentheses{
    \sum_{j=1}^{N-1} \Gamma_{j}(\Gamma_{N} + \Gamma_{j}) \Im(q_{j}^{*}\d{q}_{j})
    + \sum_{\substack{1\le j, k\le N-1\\ j\neq k}} \Gamma_{j} \Gamma_{k} \Im(q_{j}^{*}\d{q}_{k})
  }.
\end{equation*}
Let us set, with a slight abuse of notation,
\begin{equation*}
  z = (z_{1}, \dots, z_{N-2}) \defeq (q_{1} - q_{N-1}, \dots, q_{N-2} - q_{N-1}) \in \C^{N-2}
\end{equation*}
as in \eqref{eq:z}.
Notice that $z$ is in $\C^{N-2}$ as opposed to $\C^{N-1}$ here; compare with \eqref{eq:z}.
Then $z$ provides a set of coordinates for the reduced space $\mathbf{I}^{-1}(0)/\C$.
Now, $z = 0$ again exactly corresponds to $N$-tuple collisions here, and so we define
\begin{equation*}
  Z_{0} \defeq (\mathbf{I}^{-1}(0)/\C)\backslash\{\text{$N$-tuple collisions}\}
  \cong \C^{N-2} \backslash \{0\}.
\end{equation*}
We may then rewrite the above pull-back in terms of $z$ as follows:
\begin{equation*}
  i_{0}^{*}\Theta
  = -\frac{1}{2\Gamma_{N}} \parentheses{
    \sum_{j=1}^{N-2} \Gamma_{j}(\Gamma_{N} + \Gamma_{j}) \Im(z_{j}^{*}\d{z}_{j})
    + \sum_{\substack{1\le j, k\le N-2\\ j\neq k}} \Gamma_{j} \Gamma_{k} \Im(z_{j}^{*}\d{z}_{k})
  }.
\end{equation*}
Hence we have
\begin{equation*}
  i_{0}^{*}\Omega = -i_{0}^{*}\d\Theta = -\d i_{0}^{*}\Theta = \pi_{0}^{*}\Omega_{Z_{0}},
\end{equation*}
where $\pi_{0}\colon \mathbf{I}^{-1}(0) \to \mathbf{I}^{-1}(0)/\C$ is the quotient map, and $\Omega_{Z_{0}} \defeq -\d\Theta_{Z_{0}}$ with
\begin{equation}
  \label{eq:Theta_Z_0}
  \Theta_{Z_{0}} \defeq \frac{1}{2}\Im(z^{*} \mathcal{K}_{0}\,\d{z})
\end{equation}
and
\begin{equation}
  \label{eq:K_0}
  \mathcal{K}_{0} \defeq -\frac{1}{\Gamma_{N}}
  \begin{bmatrix}
    \Gamma_{1}(\Gamma_{N} + \Gamma_{1}) & \Gamma_{1} \Gamma_{2} & \dots & \Gamma_{1} \Gamma_{N-2} \\
    \Gamma_{2} \Gamma_{1} & \Gamma_{2}(\Gamma_{N} + \Gamma_{2}) & \dots & \Gamma_{2} \Gamma_{N-2} \\
    \vdots & \vdots & \ddots & \vdots \\
    \Gamma_{N-2} \Gamma_{1} & \Gamma_{N-2} \Gamma_{2} & \dots & \Gamma_{N-2}(\Gamma_{N} + \Gamma_{N-2})
  \end{bmatrix}.
\end{equation}

To summarize, we have:
\begin{proposition}
  \label{prop:Omega_Z_0}
  If $\Gamma = 0$, then the symplectic form on the $\R^{2}$-reduced space $Z_{0} \cong \C^{N-2} \backslash \{0\}$ is given by $\Omega_{Z_{0}} = -\d\Theta_{Z_{0}}$ where $\Theta_{Z_{0}}$ is the one-form defined in \eqref{eq:Theta_Z_0} along with \eqref{eq:K_0}.
\end{proposition}

\begin{remark}
  \label{rem:K-K_0}
  Comparing the matrices $\mathcal{K}$ from \eqref{eq:K} and $\mathcal{K}_{0}$ from above, one notices that the symplectic form $\Omega_{Z_{0}}$ is identical to that of $\Omega_{Z}$ for $N-1$ (as opposed to $N$) vortices with $\Gamma$ replaced by $-\Gamma_{N}$.
  That is, after the $\R^{2}$-reduction, \textit{the symplectic structure for $N$ point vortices with vanishing total circulation (i.e., $\Gamma = 0$) is the same as that for (the first) $N-1$ point vortices whose total circulation is $-\Gamma_{N} \neq 0$}.
  We note that \citet{Aref1989} observed that three-vortex motion with zero total circulation can be effectively reduced to a two-vortex problem.
  Similarly, \citet{ArSt1999} showed that four-vortex motion with zero total circulation---which is known to be integrable~\cite{Ec1988}---can be reduced to a three-vortex one as well.
\end{remark}

\section{Reduction by Rotational Symmetry}
\label{sec:ReductionByRotation}
Let us perform the further reduction by rotational symmetry.
This is the second stage of the semidirect product reduction by $\SE(2) = \SO(2) \ltimes \R^{2}$, and is more involved than that by translations.

The key ingredient is the pair of momentum maps $R$ and $\mathbf{J}$ found in the two subsections to follow:
\begin{equation}
  \label{eq:pre-dual_pair}
  \begin{tikzcd}
    \R & Z \arrow[swap]{l}{R} \arrow{r}{\mathbf{J}} & \u(\mathcal{K})^{*}.
  \end{tikzcd}
\end{equation}
The first momentum map $R$ is the conserved quantity corresponding to the $\SO(2)$-symmetry, and hence its role is clear from the point of view of symplectic reduction: The reduced space by the rotational symmetry is the Marsden--Weinstein quotient $R^{-1}(c_{0})/\mathbb{S}^{1}$ for an arbitrary regular value $c_{0} \in \R$.
The problem is that this quotient is not easy to describe and parametrize, and hence is not amenable to writing down the reduced dynamics explicitly.

Instead, we exploit the other momentum map $\mathbf{J}$, which corresponds to the natural action of the unitary group $\U(\mathcal{K})$ (see Section~\ref{ssec:U(K)_u(K)} below) on the $\R^{2}$-reduced space $Z$.
Note that this is \textit{not} a conserved quantity because $\U(\mathcal{K})$ is not a symmetry group of the system.
We show that $R$ and $\mathbf{J}$ constitute a so-called dual pair (see, e.g., \citet{We1983} and \citet[Chapter~11]{OrRa2004}) on a certain open subset of $Z$.
The dual pair helps us identify the reduced space $R^{-1}(c_{0})/\mathbb{S}^{1}$ with a coadjoint orbit in $\u(\mathcal{K})^{*}$, hence resulting in the Lie--Poisson formulation of the reduced dynamics.

Throughout the section, we will describe the results for the case with $\Gamma \neq 0$ with the symplectic manifold $Z$ and the symplectic structure $\Omega_{Z}$ defined in terms of the matrix $\mathcal{K}$.
Similar results hold for the case with $\Gamma = 0$ and $Z_{0}$ by replacing $N$ by $N-1$ and the matrix $\mathcal{K}$ by $\mathcal{K}_{0}$.

\subsection{Rotational Action on $Z$}
Let $\mathbb{S}^{1} = \setdef{ e^{\rmi\theta} \in \C }{ \theta \in [0, 2\pi) } \cong \SO(2)$ and consider the action
\begin{equation}
  \label{eq:Psi}
  \Psi\colon \mathbb{S}^{1} \times Z \to Z;
  \qquad
  \parentheses{ e^{\rmi\theta}, z = (z_{1}, \dots, z_{N-1}) }
  \mapsto \parentheses{ e^{\rmi\theta}z_{1}, \dots, e^{\rmi\theta}z_{N-1} }.
\end{equation}
This is the rotational action induced on $Z$ by the $\SE(2)$ action defined in \eqref{eq:SE2-action} after the translational $\R^{2}$-reduction performed above.
The one-form $\Theta_{Z}$ is clearly invariant under the action $\Psi$ and hence so is the symplectic form $\Omega_{Z}$ obtained in Proposition~\ref{prop:Omega_Z}, i.e., $\Psi_{e^{\rmi\theta}}^{*}\Theta_{Z} = \Theta_{Z}$ and hence $\Psi_{e^{\rmi\theta}}^{*}\Omega_{Z} = \Omega_{Z}$ for any $e^{\rmi\theta} \in \mathbb{S}^{1}$.

The corresponding infinitesimal generator is defined for any $\omega \in \so(2) \cong \R$ as follows:
\begin{equation*}
  \omega_{Z}(z)
  \defeq \left.\od{}{s} \Psi_{\exp(\rmi s\omega)}(z) \right|_{s=0}
  = \rmi\,\omega \sum_{j=1}^{N-1} \parentheses{ z_{j} \pd{}{z_{j}} - z_{j}^{*} \pd{}{z_{j}^{*}} }.
\end{equation*}
Hence the corresponding momentum map is $R\colon Z \to \R$ defined as
\begin{align*}
  R(z) \omega &= \ip{ \Theta_{Z}(z) }{ \omega_{Z}(z) } \\
  &= \frac{\omega}{2} \Im(\rmi z^{*}\mathcal{K} z) \\
  &= -\frac{\omega}{2} z^{*}\mathcal{K} z
\end{align*}
for any $\omega \in \so(2) \cong \R$. 
Therefore, we have
\begin{equation}
  \label{eq:R}
  R(z) = -\frac{1}{2} z^{*}\mathcal{K} z.
\end{equation}
Since our system has $\mathbb{S}^{1}$-symmetry, $R$ is a conserved quantity of the dynamics.
In fact, this is the so-called \textit{angular impulse}; see, e.g., \citet[Section~2.1]{Ne2001} and \citet{Ar2007}.

\subsection{Lie Group $\U(\mathcal{K})$ and Lie Algebra $\u(\mathcal{K})$}
\label{ssec:U(K)_u(K)}
Let us define a Lie group $\U(\mathcal{K})$ that naturally acts on $Z$ symplectically; then the other leg $\mathbf{J}$ of the dual pair follows from this action.
This subsection essentially reproduces the treatment of the vortex algebra of \citet{BoBoMa1999}.
The difference is that our group acts on the $\R^{2}$-reduced space $Z$ (or $Z_{0}$ if $\Gamma = 0$) whereas theirs acts on the original configuration space $\C^{N}$.
This difference stems from the fact we perform $\R^{2}$-reduction first whereas they perform $\SO(2)$-reduction first; see Section~\ref{ssec:outline} for the reason why we prefer to do so.

Let us define the Lie group
\begin{equation*}
  \U(\mathcal{K}) \defeq \setdef{ U \in \C^{(N-1)\times(N-1)} }{ U^{*} \mathcal{K} U = \mathcal{K} }.
\end{equation*}
It acts on $Z$ as follows:
\begin{equation}
  \label{eq:Phi}
  \Phi\colon \U(\mathcal{K}) \times Z \to Z;
  \qquad
  (U, z) \mapsto U z.
\end{equation}
Clearly $\Phi$ leaves the one-form $\Theta_{Z}$ invariant and hence is symplectic with respect to the symplectic form $\Omega_{Z}$.

The Lie algebra of $\U(\mathcal{K})$ is given by
\begin{equation*}
  \u(\mathcal{K}) \defeq \setdef{ \tilde{\xi} \in \C^{(N-1)\times(N-1)} }{ \tilde{\xi}^{*} \mathcal{K} + \mathcal{K} \tilde{\xi} = 0 }.
\end{equation*}
In what follows, we will not directly work with $\u(\mathcal{K})$ because it turns out to be more convenient to instead work with the Lie algebra 
\begin{equation*}
  \mathfrak{v}_{\mathcal{K}} \defeq \setdef{ \xi \in \C^{(N-1)\times(N-1)} }{ \xi^{*} = -\xi }
\end{equation*}
equipped with the non-standard Lie bracket
\begin{equation}
  \label{eq:Lie_bracket-v_K}
  [\xi, \eta]_{\mathcal{K}} \defeq \xi \mathcal{K}^{-1} \eta - \eta \mathcal{K}^{-1} \xi.
\end{equation}
In fact, we see that the map
\begin{equation}
  \label{eq:u(K)-v_K}
  \u(\mathcal{K}) \to \mathfrak{v}_{\mathcal{K}};
  \qquad
  \tilde{\xi} \mapsto \mathcal{K}\tilde{\xi} \eqdef \xi
\end{equation}
is a Lie algebra isomorphism.
Hence we will use $\u(\mathcal{K})$ and $\mathfrak{v}_{\mathcal{K}}$ interchangeably in what follows.
Note that, as a vector space, $\mathfrak{v}_{\mathcal{K}}$ is a subspace of $\u(N-1)$, but is not a subalgebra of $\u(N-1)$ in general.
\begin{remark}
  Under certain conditions on the circulations $\{ \Gamma_{j} \in \R\backslash\{0\} \}_{j=1}^{N}$, one can show that $\mathfrak{v}_{\mathcal{K}}$ is isomorphic to $\u(N-1)$; see \citet[Proposition~4]{BoBoMa1999}.
\end{remark}

Given an arbitrary $\tilde{\xi} \in \u(\mathcal{K})$, its infinitesimal generator is given by
\begin{equation*}
  \tilde{\xi}_{Z}(z) \defeq \left.\od{}{s}\right|_{s=0} \Phi_{\exp(s\tilde{\xi})}(z) = \tilde{\xi} z .
\end{equation*}
Alternatively, given an arbitrary $\xi \in \mathfrak{v}_{\mathcal{K}}$, one defines its infinitesimal generator by
\begin{equation*}
  \xi_{Z}(z) \defeq \mathcal{K}^{-1}\xi z.
\end{equation*}

What is the corresponding momentum map $\mathbf{J}\colon Z \to \u(\mathcal{K})^{*} \cong \mathfrak{v}_{\mathcal{K}}^{*}$?
First equip $\mathfrak{v}_{\mathcal{K}}$ with the inner product $\ip{\,\cdot\,}{\,\cdot\,}\colon \mathfrak{v}_{\mathcal{K}} \times \mathfrak{v}_{\mathcal{K}} \to \R$ by
\begin{equation*}
  \ip{\xi}{\eta} \defeq \frac{1}{2}\tr(\xi^{*}\eta),
\end{equation*}
and identify $\mathfrak{v}_{\mathcal{K}}^{*}$ with $\mathfrak{v}_{\mathcal{K}}$ via the inner product.
Let $\xi \in \mathfrak{v}_{\mathcal{K}}$ be arbitrary.
Then the momentum map $\mathbf{J}\colon Z \to \mathfrak{v}_{\mathcal{K}}^{*}$ is defined by
\begin{align*}
  \ip{\mathbf{J}(z)}{\xi}
  &= \ip{\Theta_{Z}(z)}{\xi_{Z}(z)} \\
  &= \frac{1}{2}\Im(z^{*} \mathcal{K} \mathcal{K}^{-1}\xi z) \\
  &= \frac{1}{2}\Im(z^{*} \xi z) \\
  &= \frac{1}{2}\tr\parentheses{ (\rmi z z^{*})^{*} \xi } \\
  &= \ip{\rmi z z^{*}}{\xi},
\end{align*}
that is,
\begin{equation}
  \label{eq:J}
  \mathbf{J}(z) = \rmi z z^{*}.
\end{equation}

We continue our treatment of $\U(\mathcal{K})$ and $\u(\mathcal{K})$---especially the associated coadjoint action and representation---in Appendix~\ref{sec:more_on_U(K)}.

\subsection{Reduction by Rotations via a Dual Pair}
\label{ssec:reduction_by_rotation}
Now that we have the pair of canonical actions $\Psi$ and $\Phi$ on $Z$ and the corresponding momentum maps $R$ and $\mathbf{J}$, the last piece of the puzzle is to identify the Marsden--Weinstein quotient $R^{-1}(c_{0})/\mathbb{S}^{1}$ with a coadjoint orbit in $\mathfrak{v}_{\mathcal{K}}^{*}$.
To that end, let us prove two lemmas that are essential for our purpose:
\begin{lemma}
  \label{lem:level_set-J}
  Each level set of $\mathbf{J}$ is an $\mathbb{S}^{1}$-orbit, i.e., for any $z \in Z$, $\mathbf{J}^{-1}(\mathbf{J}(z)) = \mathbb{S}^{1} \cdot z$.
\end{lemma}

\begin{proof}
  Let $z \in Z$ be arbitrary, and let us show that $\mathbf{J}^{-1}(\mathbf{J}(z)) \subset \mathbb{S}^{1} \cdot z$.
  First observe that, in view of \eqref{eq:J},
  \begin{equation*}
    \mathbf{J}^{-1}(\mathbf{J}(z))
    = \setdef{ w \in Z }{ w w^{*} = z z^{*} }.
  \end{equation*}
  Hence if $w \in \mathbf{J}^{-1}(\mathbf{J}(z))$ then $w w^{*} = z z^{*}$; but then it implies that $|w_{j}| = |z_{j}|$ for any $j \in \mathcal{I} \defeq \{1, \dots, N-1\}$ as well as that $w_{j} w_{k}^{*} = z_{j} z_{k}^{*}$ for any $j, k \in \mathcal{I}$ with $j \neq k$.
  The former implies that $w_{j} = e^{\rmi \theta_{j}} z_{j}$ with some $\theta_{j} \in [0, 2\pi)$ for any $j \in \mathcal{I}$.
  Now, let
  \begin{equation*}
    \mathcal{I}_{0} \defeq \setdef{ j \in \mathcal{I} }{ z_{j} = 0 }.
  \end{equation*}
  If $j \in \mathcal{I}_{0}$, then $z_{j} = 0$ and thus it follows that $w_{j} = 0$.
  On the other hand, for any $j,k \in \mathcal{I} \backslash \mathcal{I}_{0}$ with $j \neq k$, the equality $w_{j} w_{k}^{*} = z_{j} z_{k}^{*}$ implies $e^{\rmi\theta_{j}} = e^{\rmi\theta_{k}}$.
  Therefore, for any $j \in \mathcal{I} \backslash \mathcal{I}_{0}$ we have $w_{j} = e^{\rmi\theta} z_{j}$ for some $\theta \in [0, 2\pi)$; in fact, this equality is trivially satisfied for any $j \in \mathcal{I}_{0}$ as well.
  As a result, we have $w = e^{\rmi\theta} z$, i.e., $w \in \mathbb{S}^{1} \cdot z$.
  Hence we have $\mathbf{J}^{-1}(\mathbf{J}(z)) \subset \mathbb{S}^{1} \cdot z$.
  The other inclusion $\mathbb{S}^{1} \cdot z \subset \mathbf{J}^{-1}(\mathbf{J}(z))$ is trivial.
\end{proof}

\begin{lemma}
  \label{lem:level_set-R}
  Each non-zero level set of $R$ is a $\U(\mathcal{K})$-orbit, i.e., for any $z \in Z\backslash R^{-1}(0)$, $R^{-1}(R(z)) = \U(\mathcal{K}) \cdot z$.
\end{lemma}

\begin{proof}
  See Appendix~\ref{sec:level_set-R}.
\end{proof}

These results imply that we may identify the Marsden--Weinstein quotient $R^{-1}(c_{0})/\mathbb{S}^{1}$ for $c_{0} \neq 0$ with a coadjoint orbit $\mathcal{O}_{\mu_{0}}$ in $\u(\mathcal{K})^{*} \cong \mathfrak{v}_{\mathcal{K}}^{*}$ equipped with the $(+)$-Kirillov--Kostant--Souriau (KKS) symplectic structure, i.e., for any $\mu \in \mathcal{O}_{\mu_{0}}$ and $\xi, \eta \in \u(\mathcal{K}) \cong \mathfrak{v}_{\mathcal{K}}$,
\begin{equation}
  \label{eq:KKS}
  \Omega_{\mathcal{O}_{\mu_{0}}}(\mu)(-\ad_{\xi}^{*}\mu, -\ad_{\eta}^{*}\mu) \defeq \ip{\mu}{[\xi,\eta]_{\mathcal{K}}},
\end{equation}
where $[\,\cdot\,,\,\cdot\,]_{\mathcal{K}}$ is the Lie bracket on $\mathfrak{v}_{\mathcal{K}}$ defined in \eqref{eq:Lie_bracket-v_K}; see, e.g., \citet[Chapter~1]{Ki2004} and \citet[Chapter~14]{MaRa1999} and references therein.
More specifically, we have the following:
\begin{theorem}[Further reduction by rotational symmetry]
  \label{thm:rotational_reduction}
  Let $z_{0} \in Z\backslash R^{-1}(0)$ and set $c_{0} \defeq R(z_{0}) \neq 0$.
  Then the reduced space by rotational symmetry, i.e., the Marsden--Weinstein quotient $R^{-1}(c_{0})/\mathbb{S}^{1}$, is symplectomorphic to the coadjoint orbit $\mathcal{O}_{\mu_{0}} \subset \mathfrak{v}_{\mathcal{K}}^{*}$ through $\mu_{0} \defeq \mathbf{J}(z_{0}) \in \mathfrak{v}_{\mathcal{K}}^{*}$, i.e., there exists a diffeomorphism $\overline{\mathbf{J}}\colon R^{-1}(c_{0})/\mathbb{S}^{1} \to \mathcal{O}_{\mu_{0}}$ such that the diagram
  \begin{equation*}
    \begin{tikzcd}[column sep=7ex, row sep=7ex]
      Z\backslash R^{-1}(0) & \\
      R^{-1}(c_{0}) \arrow{u}{i_{c_{0}}} \arrow[swap]{d}{\pi_{c_{0}}} \arrow{dr}{\mathbf{J}|_{R^{-1}(c_{0})}} & \\
      R^{-1}(c_{0})/\mathbb{S}^{1} \arrow[swap]{r}{\overline{\mathbf{J}}} & \mathcal{O}_{\mu_{0}}
    \end{tikzcd}
  \end{equation*}
  commutes as well as that $\overline{\mathbf{J}}^{*} \Omega_{\mathcal{O}_{\mu_{0}}} = \Omega_{c_{0}}$, where $\Omega_{\mathcal{O}_{\mu_{0}}}$ is the $(+)$-KKS structure~\eqref{eq:KKS} on $\mathcal{O}_{\mu_{0}}$, and $\Omega_{c_{0}}$ is the reduced symplectic form on $R^{-1}(c_{0})/\mathbb{S}^{1}$, i.e., $i_{c_{0}}^{*} \Omega_{Z} = \pi_{c_{0}}^{*}\Omega_{c_{0}}$ with the inclusion $i_{c_{0}}\colon R^{-1}(c_{0}) \hookrightarrow Z\backslash R^{-1}(0)$ and the quotient map $\pi_{c_{0}}\colon R^{-1}(c_{0}) \to R^{-1}(c_{0})/\mathbb{S}^{1}$.
\end{theorem}

\begin{proof}
  The left half of the diagram and the relationship $i_{c_{0}}^{*} \Omega_{Z} = \pi_{c_{0}}^{*}\Omega_{c_{0}}$ are from the symplectic reduction of \citet{MaWe1974} (see also \cite[Sections~1.1 and 1.2]{MaMiOrPeRa2007}).
  
  The existence of the symplectomorphism $\bar{\mathbf{J}}$ and the commutativity of the triangle in the diagram follow from \citet[Theorem 2.9~(iii)]{BaWu2012} (see also \citet[Proposition~3.5]{Sk2018}) under the following conditions:
  (i)~The $\mathbb{S}^{1}$-action~$\Psi$ and the $\U(\mathcal{K})$-action $\Phi$ commute, (ii)~$\Psi$ and $\Phi$ are canonical actions in the sense that $\Psi^{*}\Omega_{Z} = \Omega_{Z}$ and $\Phi^{*}\Omega_{Z} = \Omega_{Z}$, (iii)~the momentum maps $R$ and $\mathbf{J}$ are equivariant, and (iv)~each level set of $\mathbf{J}$ is an $\mathbb{S}^{1}$-orbit, and each level set of $R$ is a $\U(\mathcal{K})$-orbit.

  Note that, due to the result of Lemma~\ref{lem:level_set-R}, we first restrict the definitions of the actions $\Psi$ and $\Phi$ and the momentum maps $R$ and $\mathbf{J}$ to the open subset $Z\backslash R^{-1}(0)$; we do not change the notation to avoid unnecessary complications.
  Then, (i) and (ii) are clear from the definitions \eqref{eq:Psi} and \eqref{eq:Phi} of $\Psi$ and $\Phi$ as well as that of the symplectic form $\Omega_{Z}$ in Proposition~\ref{prop:Omega_Z}; (iii) is also clear from the definitions \eqref{eq:R} and \eqref{eq:J} of the momentum maps; (iv) follows from Lemmas~\ref{lem:level_set-J} and \ref{lem:level_set-R} from above.
\end{proof}

\begin{remark}
  Clearly, both $\Psi$ and $\Phi$ are free; note that $Z \defeq \C^{N-1} \backslash \{0\}$.
  Then the conditions we checked above implies (see \citet[Proposition~3.7]{Sk2018}) that the momentum maps $R$ and $\mathbf{J}$ form a dual pair on $Z\backslash R^{-1}(0)$ in the sense of \citet{We1983} (see also \citet[Chapter~11]{OrRa2004}), i.e., the pair of Poisson maps~\eqref{eq:pre-dual_pair} satisfies $(\ker T_{z}R)^{\Omega_{Z}} = \ker T_{z}\mathbf{J}$ for any $z \in Z\backslash R^{-1}(0)$.
\end{remark}

\subsection{Lie--Poisson Equation for Reduced Dynamics}
Theorem~\ref{thm:rotational_reduction} implies that the dynamics of $N$ point vortices with non-zero circulations defined by \eqref{eq:N-point_vortices} is reduced to a Lie--Poisson equation on $\u(\mathcal{K})^{*} \cong \mathfrak{v}_{\mathcal{K}}^{*}$.
More specifically, we have the following:
\begin{corollary}[Reduced dynamics of $N$ point vortices]
  \label{cor:reduced_dynamics}
  Consider the dynamics of $N$ point vortices with non-zero circulations $\{ \Gamma_{j} \in \R\backslash\{0\} \}_{j=1}^{N}$ defined by \eqref{eq:N-point_vortices}.
  Suppose that the total circulation is non-zero, i.e., $\Gamma\defeq \sum_{j=1}^{N}\Gamma_{j} \neq 0$, and let $\mathbf{q}(0) \in \C^{N}$ be the initial condition for \eqref{eq:N-point_vortices}, $z_{0} \in Z$ be the corresponding element defined by \eqref{eq:z}, and $\mu_{0} \defeq \mathbf{J}(z_{0})$.
  If $R(z_{0}) \neq 0$ (i.e., the angular impulse is non-zero), then:
  \begin{enumerate}[(i)]
  \item The $\SE(2)$-reduced dynamics in the coadjoint orbit $\mathcal{O}_{\mu_{0}}$ is described by $\mu = \mathbf{J}(z)$ satisfying the Lie--Poisson equation
    \begin{equation}
      \label{eq:Lie-Poisson}
      \dot{\mu} = -\ad_{\delta h/\delta\mu}^{*}\mu,
    \end{equation}
    where $h\colon \mathfrak{v}_{\mathcal{K}}^{*} \to \R$ is a collective Hamiltonian, i.e., $H_{Z} = h \circ \mathbf{J}$.
  \item In addition to the Hamiltonian $h$, the Casimirs $\setdef{C_{j}\colon \mathfrak{v}_{\mathcal{K}}^{*} \to \R}{ j \in \N }$ defined by
    \begin{equation}
      \label{eq:Casimirs}
      C_{j}(\mu) \defeq \tr((\rmi\,\mathcal{K} \mu)^{j})
    \end{equation}
    are conserved in the reduced dynamics.
  \end{enumerate}
\end{corollary}

\begin{proof}
  (i)~It is a direct consequence of Theorem~\ref{thm:rotational_reduction}.
  (ii)~Clearly the Hamiltonian $h$ is conserved.
  As for the Casimirs, note first that the expression~\eqref{eq:Adstar-v_K} for the coadjoint action of $\U(\mathcal{K})$ on $\mathfrak{v}_{\mathcal{K}}^{*}$ suggests that the functions $\{ C_{j} \}_{j\in \N}$ are all $\Ad^{*}$-invariant, i.e., $C_{j}(\Ad_{U^{-1}}^{*} \mu) = C_{j}(\mu)$ for any $\mu \in \mathfrak{v}_{\mathcal{K}}^{*}$ as verified easily.
  Since any $\Ad^{*}$-invariant differentiable function is a Casimir of the $(+)$-Lie--Poisson bracket (see, e.g., \cite[Corollary~14.4.3]{MaRa1999})
  \begin{equation*}
    \PB{f}{h}(\mu) \defeq \ip{\mu}{ \brackets{ \frac{\delta f}{\delta\mu}, \frac{\delta h}{\delta\mu} }_{\mathcal{K}} }
  \end{equation*}
  on $\mathfrak{v}_{\mathcal{K}}^{*}$, we conclude that $\{ C_{j} \}_{j\in \N}$ are Casimirs and hence are conserved quantities of the reduced dynamics.
\end{proof}

\begin{remark}
  A more concrete expression for the Lie--Poisson equation~\eqref{eq:Lie-Poisson} is, using \eqref{eq:adstar-v_K},
  \begin{equation}
    \label{eq:Lie-Poisson-matrix_form}
    \dot{\mu} = -\ad_{\delta h/\delta\mu}^{*}\mu = -\mu \frac{\delta h}{\delta\mu} \mathcal{K}^{-1} + \mathcal{K}^{-1} \frac{\delta h}{\delta\mu} \mu,
  \end{equation}
  where the derivative ${\delta h}/{\delta\mu} \in \mathfrak{v}_{\mathcal{K}}$ is defined so that, for any $\mu,\nu \in \mathfrak{v}_{\mathcal{K}}^{*}$,
  \begin{equation*}
    \ip{\nu}{\frac{\delta h}{\delta\mu}}
    = \frac{1}{2}\tr\parentheses{ \nu^{*} \frac{\delta h}{\delta\mu} }
    = \left.\od{}{s}\right|_{s=0} h(\mu + s\nu).
  \end{equation*}
\end{remark}

\begin{remark}
  \label{rem:C_1-R}
  The Casimir $C_{1}$ is essentially the angular impulse $R$.
  In fact, we have
  \begin{equation*}
    C_{1} \circ \mathbf{J}(z) = \tr(-\mathcal{K} z z^{*}) = -z^{*} \mathcal{K} z = 2 R(z).
  \end{equation*}
\end{remark}

\begin{remark}
  As mentioned in the beginning of the section, the results of both Theorem~\ref{thm:rotational_reduction} and Corollary~\ref{cor:reduced_dynamics} apply to the case with vanishing total circulation by replacing $N$ by $N-1$ and $\mathcal{K}$ by $\mathcal{K}_{0}$.
\end{remark}

\section{Back to the Examples}
Now we would like to apply the above results to the motivating examples from Section~\ref{ssec:motivating_examples}.
We show that the shape dynamics in these examples are indeed periodic exploiting the Lie--Poisson formulation and the Casimirs found above.

\subsection{\boldmath $N = 3$ with $\Gamma \neq 0$}
\label{ssec:N=3}
  We may write the elements in $\mathfrak{v}_{\mathcal{K}}$ as
  \begin{equation*}
    \mathfrak{v}_{\mathcal{K}} = \setdef{
      \rmi
      \begin{bmatrix}
        \mu_{2} & \mu_{3} + \rmi\,\mu_{4} \\
        \mu_{3} - \rmi\,\mu_{4} & \mu_{1}
      \end{bmatrix}
    }{ \mu_{1}, \mu_{2}, \mu_{3}, \mu_{4} \in \R },
  \end{equation*}
  which can be identified with $\R^{4} = \{ (\mu_{1}, \mu_{2}, \mu_{3}, \mu_{4}) \}$.
  By setting $\mu = \mathbf{J}(z)$, we have
  \begin{gather*}
    \mu_{1} = |z_{2}|^{2} = | q_{2} - q_{3} |^{2},
    \qquad
    \mu_{2} = |z_{1}|^{2} = | q_{1} - q_{3} |^{2},
    \\
    \mu_{3} + \rmi\,\mu_{4} = z_{1} z_{2}^{*} = ( q_{1} - q_{3} )( q_{2}^{*} - q_{3}^{*} ).
  \end{gather*}
  The derivative $\delta h/\delta\mu$ is then
  \begin{equation*}
    \frac{\delta h}{\delta\mu}
    = \rmi
    \begin{bmatrix}
      2 \tpd{h}{\mu_{2}} & \tpd{h}{\mu_{3}} + \rmi\,\tpd{h}{\mu_{4}} \\
      \tpd{h}{\mu_{3}} - \rmi\,\tpd{h}{\mu_{4}} & 2 \tpd{h}{\mu_{1}}
    \end{bmatrix}
    = \parentheses{
      2\pd{h}{\mu_{1}}, 2\pd{h}{\mu_{2}}, \pd{h}{\mu_{3}}, \pd{h}{\mu_{4}}
    }.
  \end{equation*}
  We define the collective Hamiltonian $h$ as
  \begin{equation*}
    h(\mu) \defeq -\frac{1}{4\pi} \parentheses{
      \Gamma_{1} \Gamma_{2} \ln(\mu_{1} + \mu_{2} - 2\mu_{3})
      + \Gamma_{2} \Gamma_{3} \ln \mu_{1}
      + \Gamma_{3} \Gamma_{1} \ln \mu_{2}
    }.
  \end{equation*}

  The Lie--Poisson equation~\eqref{eq:Lie-Poisson} or \eqref{eq:Lie-Poisson-matrix_form} then gives
  \begin{equation*}
    \begin{array}{c}
      \DS\dot{\mu}_{1} = \frac{\Gamma_{1}}{\pi}\,f_{1}(\mu)\,\mu_{4},
      \qquad
      \dot{\mu}_{2} = \frac{\Gamma_{2}}{\pi}\,f_{2}(\mu)\,\mu_{4},
      \qquad
      \DS\dot{\mu}_{3} = \frac{1}{2 \pi} \left( \sum_{j=1}^{3} \Gamma_{j} f_{j}(\mu) \right) \mu_{4},
      \medskip\\
      \DS\dot{\mu}_{4} = -\frac{1}{2 \pi } \left(
      \Gamma_{1} f_{1}(\mu)(\mu_{3}-\mu_{2}) 
        + \Gamma_{2} f_{2}(\mu)(\mu_{3}-\mu_{1}) 
      + \Gamma_{3} f_{3}(\mu)\mu_{3}
        \right),
    \end{array}
  \end{equation*}
  where
  \begin{equation*}
    f_{1}(\mu) \defeq \frac{1}{\mu_{1}+\mu_{2}-2 \mu_{3}}-\frac{1}{\mu_{2}},
    \qquad
    f_{2}(\mu) \defeq \frac{1}{\mu_{1}}-\frac{1}{\mu_{1}+\mu_{2}-2 \mu_{3}},
    \qquad
    f_{3}(\mu) \defeq \frac{1}{\mu_{1}}-\frac{1}{\mu_{2}}.
  \end{equation*}

  The linear Casimir $C_{1}$ (essentially the angular impulse $R$; see Remark~\ref{rem:C_1-R}) is written in terms of $\mu$ as follows:
  \begin{equation*}
    C_{1}(\mu) = \frac{\Gamma_{2} (\Gamma_{1}+\Gamma_{3}) \mu_{1} + \Gamma_{1} (\Gamma_{2}+\Gamma_{3}) \mu_{2} -2 \Gamma_{1} \Gamma_{2} \mu_{3}}{\Gamma_{1}+\Gamma_{2}+\Gamma_{3}}.
  \end{equation*}
  It is easy to see that the three conserved quantities---the Hamiltonian $h$, the linear and quadratic Casimirs $C_{1}$ and $C_{2}$ (see \eqref{eq:Casimirs})---are independent.

  The variables $\mu = (\mu_{1}, \mu_{2}, \mu_{3}, \mu_{4})$ are related to the inter-vortex distance $l_{ij} \defeq |q_{i} - q_{j}|$ and the signed area $A$ of the triangle introduced in Section~\ref{ssec:motivating_examples} as follows:
  \begin{align*}
    (l_{23}^{2}, l_{31}^{2}, l_{12}^{2}, A)
    &= \parentheses{ \mu_{1}, \mu_{2}, \mu_{1}+\mu_{2}-2 \mu_{3}, -\frac{1}{2}\mu_{4}} \nonumber\\
    &= \parentheses{
      | q_{2} - q_{3} |^{2},
      | q_{1} - q_{3} |^{2},
      | q_{1} - q_{2} |^{2},
      -\frac{1}{2}\Im( (q_{1} - q_{3})(q_{2}^{*} - q_{3}^{*}))
      }.
  \end{align*}
  Rewriting the the Lie--Poisson equation and the Casimir $C_{2}$ in the new variables, we obtain the equations~\eqref{eq:relative_motion} of relative motion as well as the expression~\eqref{eq:C_2} for the Casimir $C_{2}$.

  Since $C_{1}$ is linear in $(\mu_{1},\mu_{2},\mu_{3})$, its level set $C_{1}^{-1}(2c_{0})$ defines an affine subspace of codimension 1 in $\mathfrak{v}(\mathcal{K})$; hence we may parametrize the level set of $C_{1}^{-1}(2c_{0})$ by $(\mu_{1},\mu_{2},\mu_{4})$.
  One may then restrict the collective Hamiltonian $h$ and the quadratic Casimir $C_{2}$ in this affine subspace.
  Then the Lie--Poisson dynamics is in the intersection of the level sets of $h$ and $C_{2}$ in the affine subspace $C_{1}^{-1}(2c_{0})$.
  Figure~\ref{fig:Lie-Poisson}~(a) shows the Lie--Poisson shape dynamics of three point vortices with the parameters and initial condition specified in \eqref{eq:3PVs}.
  The shape dynamics is in the one-dimensional manifold defined as the intersection of the Casimir $C_{2}$ and the Hamiltonian.
  This confirms the periodicity of the shape dynamics alluded in Section~\ref{ssec:motivating_examples}.
  
  \begin{figure}[htbp]
    \centering
    \subfigure[$N = 3$ with \eqref{eq:3PVs}; see Section~\ref{ssec:N=3}.
    The level set of the quadratic Casimir $C_{2}$ (green) defines an ellipsoid.]{
      \includegraphics[width=.465\linewidth]{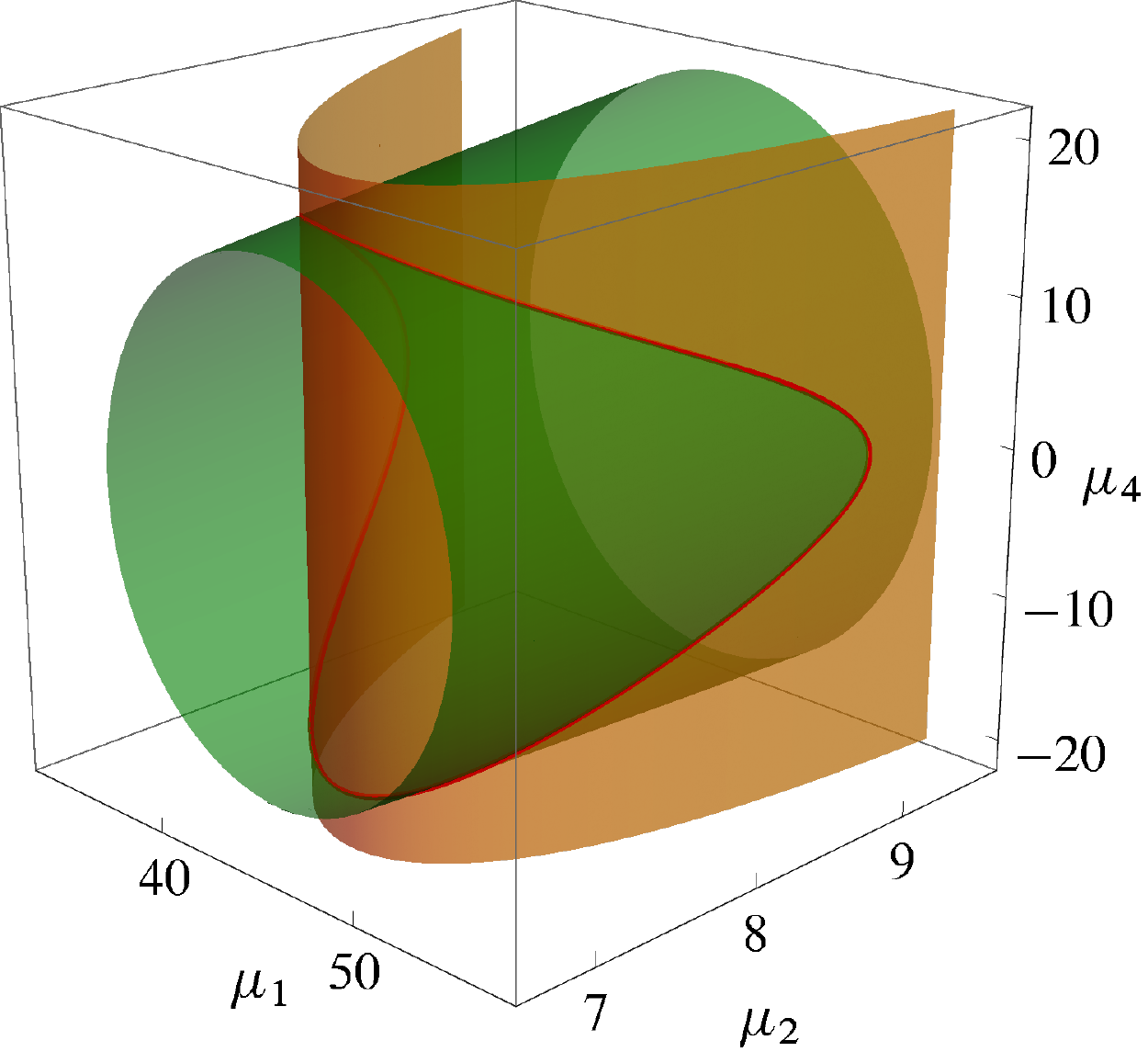}
    }
    \quad
    \subfigure[$N = 4$ with \eqref{eq:4PVs}; see Section~\ref{ssec:N=4}.
    The level set of the quadratic Casimir $C_{2}$ (green) defines a paraboloid.]{
      \includegraphics[width=.465\linewidth]{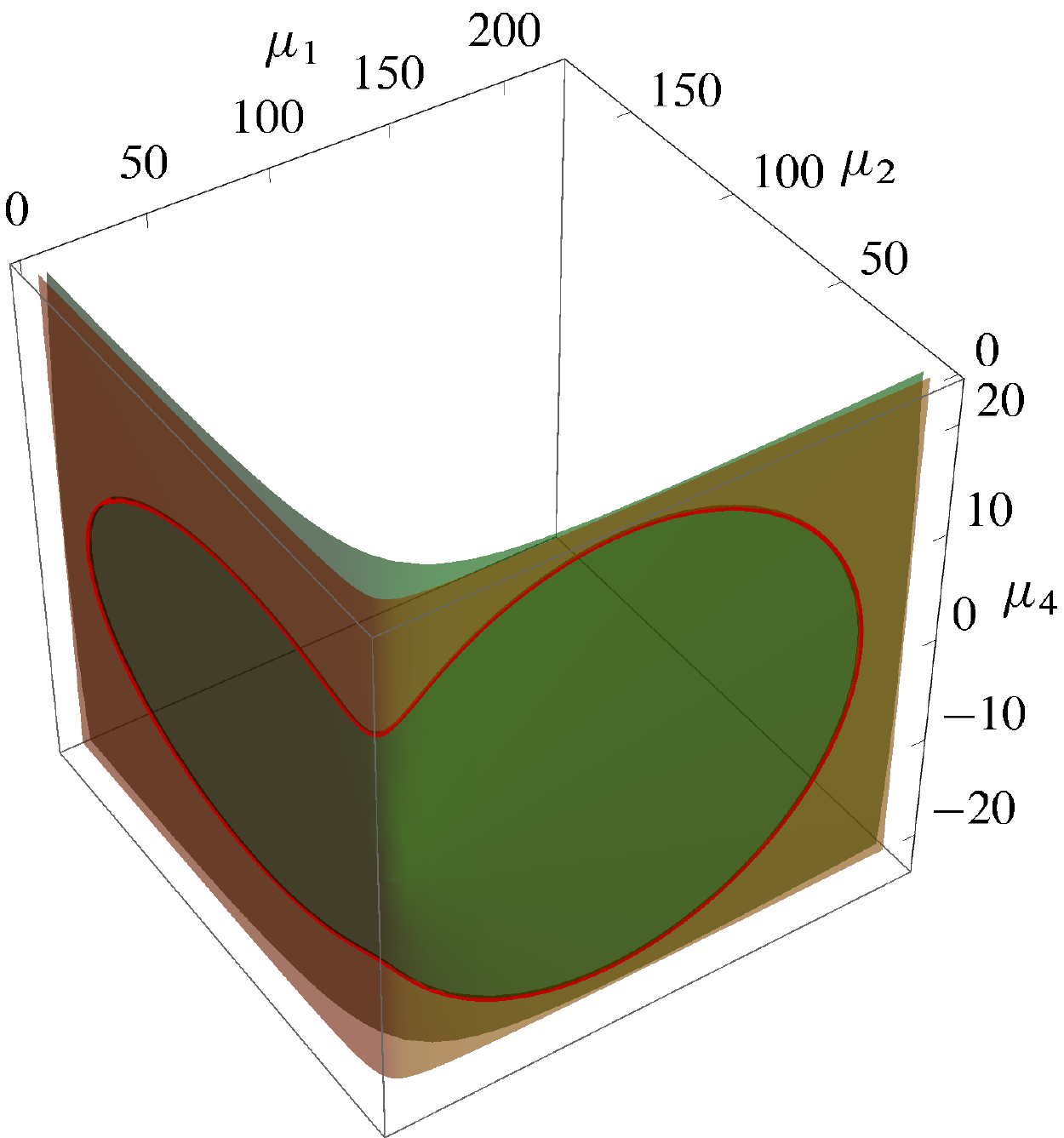}
    }
    \captionsetup{width=0.95\textwidth}
    \caption{
      Lie--Poisson dynamics shape dynamics (red) of point vortices for (a)~$N = 3$ with \eqref{eq:3PVs} and (b)~$N = 4$ with \eqref{eq:4PVs} from Section~\ref{ssec:motivating_examples}.
      The level set of the linear Casimir $C_{1}$ defines an affine subspace of $\mathfrak{v}_{\mathcal{K}} = \{(\mu_{1}, \mu_{2}, \mu_{3}, \mu_{4})\}$ with codimension 1, and hence can be parametrized by $(\mu_{1}, \mu_{2}, \mu_{4})$.
      The green and orange surfaces are the level sets of the quadratic Casimir $C_{2}$ and Hamiltonian $h$, respectively, in $\R^{3} = \{(\mu_{1}, \mu_{2}, \mu_{4})\}$.
    }
    \label{fig:Lie-Poisson}
  \end{figure}

\subsection{\boldmath $N = 4$ with $\Gamma = 0$}
\label{ssec:N=4}
Let us next consider the four-vortex case with zero total circulation $\Gamma$.
Just like the general three-vortex case, this is an integrable case as well; see \citet{Ec1988}.

As discussed in Propositions~\ref{prop:translational_reduction} and \ref{prop:Omega_Z_0} (see also Remark~\ref{rem:K-K_0}), the $\R^{2}$-reduced space $Z_{0}$ in this case is $\C^{2} \backslash \{0\}$, and so the Lie algebra $\mathfrak{v}_{\mathcal{K}_{0}}$ is essentially the same as $\mathfrak{v}_{\mathcal{K}}$ from Section~\ref{ssec:N=3} with $N = 3$.
Hence one can formulate the Lie--Poisson dynamics as well as demonstrate the periodicity of the shape dynamics just as in the above example; see Fig.~\ref{fig:Lie-Poisson}~(b).

\section{Conclusion and Outlook}
\subsection{Conclusion}
We applied symplectic reduction by $\SE(2)$ to the Hamiltonian dynamics of $N$ point vortices on the plane, and came up with the Lie--Poisson formulation of the shape dynamics of the vortices.
As stated in the introduction, the approach is similar to that of \citet{BoBoMa1999} and the result is essentially the same, but our work clarifies how the symplectic/Poisson structure of the shape dynamics is inherited by applying symplectic reduction by stages, first by $\R^{2}$ and then by $\SO(2)$.
The second stage uses a dual pair; this facilitates the reduction and naturally gives rise to the Lie--Poisson structure.
We also found a family of Casimirs of the Lie--Poisson structure.
Some of those Casimirs apparently have been overlooked in the existing literature.
The examples provided demonstrate the use of the Lie--Poisson formulation and the Casimirs to show that some shape dynamics are periodic although the trajectories of the vortices on the plane are not.
The non-periodicity of the trajectories is due to the reconstruction phase picked up by the full dynamics when the reduced/shape dynamics undergoes a periodic motion; see \citet{HeSh2018} for more details.
It is an interesting future work to investigate the reconstruction phase using our geometric setting and the Lie--Poisson formulation.

\subsection{Outlook}
As illustrated in the $\SO(2)$-reduction in this paper, a dual pair facilitates a symplectic reduction by deducing that the reduced dynamics is a Lie--Poisson dynamics.
This approach in general is particularly useful if the Marsden--Weinstein quotient turns out to be complicated: One can sidestep the difficulty by formulating the reduced dynamics as as Lie--Poisson equation, which is defined on a vector space.
In other words, a dual pair significantly simplifies the description of the seemingly complicated reduced dynamics; this facilitates numerical computations as well.

In the last few years, there have been some new developments and applications of dual pairs.
\citet{Sk2018} constructed a dual pair to give a different perspective of the realization of the Siegel upper half space as a Marsden--Weinstein quotient by the author~\cite{Oh2015c}; this work was originally motivated by the dynamics of semiclassical wave packets.
Recently, the author~\cite{Ohsawa-SymRep} also used a dual pair constructed by \citet{SkVi2019} to the symmetric representation of the rigid body equation~\cite{BlCrMaRa2002} to show that it is related to the generalized rigid body equation via a symplectic reduction.
Furthermore, the extension of this paper to the dynamics of $N$ point vortices on the sphere is under way~\cite{Ohsawa-PointPortices-Sphere}, and again we use a dual pair constructed by \citet{SkVi2019}.
The same idea may be used to analyze the dynamics of relative configurations of interacting quantum spin systems, because its geometric structure is similar to that of the $N$ point vortices on the sphere.

\appendix
\numberwithin{equation}{section}

\section{More on Lie Group $\U(\mathcal{K})$ and Lie Algebra $\u(\mathcal{K})$}
\label{sec:more_on_U(K)}
\subsection{Coadjoint Action and Casimirs}
The adjoint action $\Ad\colon \U(\mathcal{K}) \times \u(\mathcal{K}) \to \u(\mathcal{K})$ is defined as
\begin{equation*}
  \Ad_{U}\tilde{\eta} \defeq U \tilde{\eta} U^{-1}.
\end{equation*}
Since we identify $\u(\mathcal{K})$ with $\mathfrak{v}_{\mathcal{K}}$ via the map~\eqref{eq:u(K)-v_K}, the corresponding action of $U(\mathcal{K})$ on $\mathfrak{v}_{\mathcal{K}}$ is given by, with an abuse of notation,
\begin{align*} 
  \Ad_{U}\eta
  &\defeq \mathcal{K} \Ad_{U} \tilde{\eta} \\
  &= \mathcal{K} U \mathcal{K}^{-1} \eta \mathcal{K}^{-1} U^{*} \mathcal{K} \\
  &= (U^{-1})^{*} \eta U^{-1},
\end{align*}
where we used the relation $U^{-1} = \mathcal{K}^{-1} U^{*} \mathcal{K}$.
Hence $\Ad_{U^{-1}} \eta = U^{*} \eta U$ and thus we obtain the coadjoint action of $U(\mathcal{K})$ on $\mathfrak{v}_{\mathcal{K}}^{*}$ as follows:
\begin{equation}
  \label{eq:Adstar-v_K}
  \Ad_{U^{-1}}^{*} \mu = U \mu U^{*}.
\end{equation}

\subsection{Coadjoint Representation}
From the above expression of the adjoint action on $\mathfrak{v}_{\mathcal{K}}$, we have the adjoint representation of $\u(\mathcal{K})$ on $\mathfrak{v}_{\mathcal{K}}$ as
\begin{equation*}
  \ad_{\tilde{\xi}} \eta = -\tilde{\xi}^{*} \eta - \eta \tilde{\xi}
\end{equation*}
Again we abuse the notation and define the adjoint representation of $\mathfrak{v}_{\mathcal{K}}$ on itself as
\begin{align*}
  \ad_{\xi} \eta \defeq \ad_{\tilde{\xi}} \eta
  &= \xi \mathcal{K}^{-1} \eta - \eta \mathcal{K}^{-1} \xi,
\end{align*}
which coincides with the Lie bracket~\eqref{eq:Lie_bracket-v_K} on $\mathfrak{v}_{\mathcal{K}}$.
As a result, we obtain the coadjoint representation of $\mathfrak{v}_{\mathcal{K}}$ on $\mathfrak{v}_{\mathcal{K}}^{*}$ as follows:
\begin{equation}
  \label{eq:adstar-v_K}
  \ad_{\xi}^{*}\mu = \mu \xi \mathcal{K}^{-1} - \mathcal{K}^{-1} \xi \mu.
\end{equation}

\section{Proof of Lemma~\ref{lem:level_set-R}}
\label{sec:level_set-R}
\begin{lemma}
  \label{lem:detK}
  The determinant of the matrix $\mathcal{K}$ defined in \eqref{eq:K} is given by
  \begin{equation*}
    \det \mathcal{K} = \frac{(-1)^{N-1}}{\Gamma} \prod_{j=1}^{N} \Gamma_{j} = (-1)^{N-1} \frac{\Gamma_{1} \cdots \Gamma_{N}}{\Gamma_{1} + \dots + \Gamma_{N}}.
  \end{equation*}
\end{lemma}

\begin{proof}
  From the expression~\eqref{eq:K} for $\mathcal{K}$, we see that
  \begin{equation*}
    \det \mathcal{K} = \frac{1}{\Gamma^{N-1}}
    \parentheses{ \prod_{j=1}^{N-1} \Gamma_{j} }
    \begin{vmatrix}
      \Gamma_{1} - \Gamma & \Gamma_{1} & \dots & \Gamma_{1} \\
      \Gamma_{2} & \Gamma_{2} - \Gamma & \dots & \Gamma_{2} \\
      \vdots & \vdots & \ddots & \vdots \\
      \Gamma_{N-1} & \Gamma_{N-1} & \dots & \Gamma_{N-1} - \Gamma
    \end{vmatrix}.
  \end{equation*}
  However, setting $\boldsymbol{\Gamma} = (\Gamma_{1}, \dots, \Gamma_{N-1})$ and $\mathbf{1} = (1, \dots, 1)$ in $\R^{N-1}$, the determinant on the right-hand side can be written as
  \begin{align*}
    \begin{vmatrix}
      \Gamma_{1} - \Gamma & \Gamma_{1} & \dots & \Gamma_{1} \\
      \Gamma_{2} & \Gamma_{2} - \Gamma & \dots & \Gamma_{2} \\
      \vdots & \vdots & \ddots & \vdots \\
      \Gamma_{N-1} & \Gamma_{N-1} & \dots & \Gamma_{N-1} - \Gamma
    \end{vmatrix}
                 &= \det\parentheses{ \boldsymbol{\Gamma} \mathbf{1}^{T} -\Gamma I } \\
                              &= (-\Gamma)^{N-1} \det\parentheses{ I - \frac{1}{\Gamma} \boldsymbol{\Gamma} \mathbf{1}^{T} } \\
                              &= (-\Gamma)^{N-1} \parentheses{ 1 - \frac{1}{\Gamma} \boldsymbol{\Gamma}^{T} \mathbf{1} } \\
                              &= (-1)^{N-1} \Gamma^{N-2} \Gamma_{N},
  \end{align*}
  where we used the fact that $\det(I + \mathbf{x}\mathbf{y}^{T}) = 1 + \mathbf{x}^{T}\mathbf{y}$ for any $n \times n$ identity matrix $I$ and any $\mathbf{x}, \mathbf{y} \in \R^{n}$.
\end{proof}

\begin{remark}
  Similarly, we have
  \begin{equation*}
    \det \mathcal{K}_{0} = \frac{(-1)^{N-1}}{\Gamma_{N}} \prod_{j=1}^{N-1} \Gamma_{j} = (-1)^{N} \frac{\Gamma_{1} \cdots \Gamma_{N-1}}{\Gamma_{1} + \dots + \Gamma_{N-1}},
  \end{equation*}
  where $\Gamma = \sum_{j=1}^{N} \Gamma_{j} = 0$ is assumed.
  If follows easily by replacing $N$ by $N-1$ and $\Gamma$ by $-\Gamma_{N}$; see Remark~\ref{rem:K-K_0}.
\end{remark}

\begin{proof}[Proof of Lemma~\ref{lem:level_set-R}]
  It suffices to show that the Lie group $\U(\mathcal{K})$ acts transitively on the level set $R^{-1}(c)$ of the momentum map \eqref{eq:R} for any $c \in \R\backslash\{0\}$ because that implies that $R^{-1}(R(z)) \subset \U(\mathcal{K}) \cdot z$ whereas the other inclusion $\U(\mathcal{K}) \cdot z \subset R^{-1}(R(z))$ is trivial.
  
  By the assumption and the above lemma, we have $\det \mathcal{K} \neq 0$.
  Therefore, the inner product on $\C^{N-1} \supset Z$ defined by
  \begin{equation*}
    \ip{v}{w}_{\mathcal{K}} \defeq v^{*} \mathcal{K} w
  \end{equation*}
  for any $v, w \in \C^{N-1}$ is non-degenerate in the sense that $\ip{v}{w}_{\mathcal{K}} = 0$ for any $w \in \C^{N-1}$ implies that $v = 0$.
  This implies that one can find a basis for $\C^{N-1}$ with respect to which $\mathcal{K}$ is expressed as $\begin{tbmatrix} I_{n_{1}} & 0 \\ 0 & -I_{n_{2}} \end{tbmatrix}$ for some $n_{1}, n_{2} \in \{0, \dots, N-1\}$ such that $n_{1} + n_{2} = N-1$; as a result, one sees that $\U(\mathcal{K})$ is isomorphic to the indefinite unitary group (see, e.g., \citet[Lemma~1.1.7 and Proposition~1.1.8]{GoWa2009})
  \begin{equation*}
    \U(n_{1},n_{2}) \defeq \setdef{ U \in \C^{(N-1)\times(N-1)} }{ U^{*} \begin{tbmatrix}
        I_{n_{1}} & 0 \\
        0 & -I_{n_{2}}
      \end{tbmatrix} U
      = \begin{tbmatrix}
        I_{n_{1}} & 0 \\
        0 & -I_{n_{2}}
      \end{tbmatrix}
    }.
  \end{equation*}
  Then the momentum map $R$ is written as
  \begin{equation*}
    R(z) = \sum_{j=1}^{n_{1}} |z_{j}|^{2} - \sum_{k=1}^{n_{2}} |z_{n_{1}+k}|^{2}
  \end{equation*}
  in terms of the coordinates with respect to this basis.

  Let us consider the level set $R^{-1}(c)$ with $c > 0$.
  The level set may be written as
  \begin{equation*}
    R^{-1}(c) = \bigcup_{b\ge c} \mathcal{S}_{c}(b),
  \end{equation*}
  where
  \begin{equation*}
    \mathcal{S}_{c}(b) \defeq \setdef{ z \in \C^{N-1} }{ \sum_{j=1}^{n_{1}} |z_{j}|^{2} = b,\, \sum_{k=1}^{n_{2}} |z_{n_{1}+k}|^{2} = b - c }.
  \end{equation*}

  Let $b \ge c$ be arbitrary and set $w = (\tilde{w}, \hat{w}) \in \mathcal{S}_{c}(b)$ with $\tilde{w} = (\sqrt{b}, 0, \dots, 0) \in \C^{n_{1}}$ and $\hat{w} = (\sqrt{b - c}, 0, \dots, 0) \in \C^{n_{2}}$.
  Then, given any point $z = (\tilde{z}, \hat{z}) \in \mathcal{S}_{c}(b)$, one sees that $\tilde{z} \in \mathbb{S}^{2n_{1}-1}_{\sqrt{b}} \subset \C^{n_{1}}$ and $\tilde{z} \in \mathbb{S}^{2n_{2}-1}_{\sqrt{b-c}} \subset \C^{n_{2}}$; where $\mathbb{S}_{r}^{n}$ stands for the $n$-sphere with radius $r > 0$ centered at the origin.
  Therefore, one can find $\tilde{W} \in \U(n_{1})$ and $\hat{W} \in \U(n_{2})$ such that $\tilde{z} = \tilde{W} \tilde{w}$ and $\hat{z} = \hat{W} \hat{w}$.
  Then, setting $W \defeq
  \begin{tbmatrix}
    \tilde{W} & 0 \\
    0 & \hat{W}
  \end{tbmatrix}$, one sees that $W \in \U(n_{1},n_{2})$ and $z = W w$.

  Now, pick $v = (\sqrt{c}, 0, \dots, 0) \in \mathcal{S}_{c}(c)$.
  For any $b \ge c$ there exists $t_{b} \ge 0$ such that $\cosh t_{b} = \sqrt{b/c}$ and $\sinh t_{b} = \sqrt{(b - c)/c}$.
  Therefore, by setting
  \begin{equation*}
    U_{b} \defeq
    \begin{bmatrix}
      \cosh t_{b} & 0 & \sinh t_{b} & 0\\
      0 & I_{n_{1}-1} & 0 & 0 \\
      \sinh t_{b} & 0 & \cosh t_{b} & 0 \\
      0 & 0 & 0 & I_{n_{2}-1} \\
    \end{bmatrix} \in \U(n_{1},n_{2}),
  \end{equation*}
  we have $w = U_{b} v$.
  As a result, any $z \in \mathcal{S}_{c}(b)$ is written as $z = W U_{b} v$ with $W U_{b} \in \U(n_{1},n_{2})$.
  Since $b \ge c$ is arbitrary, $\U(\mathcal{K}) \cong \U(n_{1},n_{2})$ acts transitively on the level set $R^{-1}(c)$ for any $c > 0$.

  One can argue similarly for $c < 0$ as well.
\end{proof}

\section*{Acknowledgments}
I would like to thank Paul Skerritt for his helpful comments and discussions on dual pairs.

\bibliography{Point_Vortices-Plane}
\bibliographystyle{plainnat}

\end{document}